\documentclass[a4paper,USenglish]{lipics-v2016}
\usepackage{microtype}
\usepackage{authblk}
\usepackage{graphicx}
\usepackage{xspace}
\usepackage{amsmath}
\usepackage{amsthm}
\usepackage{amssymb}
\usepackage{framed,multicol}
\usepackage[algo2e]{algorithm2e}
\usepackage{algorithm,algorithmic,mathtools}
\usepackage{soul}
\usepackage{caption}
\usepackage[nameinlink]{cleveref}
\usepackage{comment}
\usepackage{multirow}
\usepackage{multicol}
\usepackage{lipsum}
\usepackage{subcaption}
\usepackage{color}
\usepackage[symbol]{footmisc}

\newcommand{\argmin}{\operatornamewithlimits{argmin}}

\newcommand{\positivereal}{\ensuremath{\mathbb{R}^+}\xspace}
\newcommand{\nonnegativereal}{{{\ensuremath{[0,\infty)}}}}

\newcommand\p{\mbox{\bf P}\xspace}
\newcommand\np{\mbox{\bf NP}\xspace}
\newcommand\nph{\mbox{{\bf NP}{\rm -hard}}\xspace}
\newcommand\ugch{\mbox{{\bf UGC}{\rm -hard}}\xspace}

\def\poly{\operatorname{poly}}

\newcommand{\minmp}{\textsc{Min Middlebox Node Purchase}\xspace}
\newcommand{\maxmp}{\textsc{Budgeted Middlebox Node Purchase}\xspace}

\newcommand{\tempcolor}{black}
\newcommand{\Twalk}{\textcolor{\tempcolor}{2-Walk}\xspace}

\newcommand{\twalk}{\textcolor{\tempcolor}{2-walk}\xspace}
\newcommand{\twalks}{\textcolor{\tempcolor}{2-walks}\xspace}

\crefname{thm}{theorem}{theorems}
\crefname{lemma}{lemma}{lemmas}

\title{Multi-Commodity Flow with In-Network Processing$^\dagger$}

\author[1]{Moses Charikar}
\author[2]{Yonatan Naamad}
\author[3]{Jennifer Rexford}
\author[4]{X. Kelvin Zou}
\affil[1]{Stanford University, Stanford, CA, USA \\
  \texttt{moses@cs.stanford.edu}}
\affil[2]{Amazon.com, Inc., Palo Alto, USA.\thanks{This work was done while the author was at the Department of Computer Science, Princeton University.}\\
  \texttt{ynaamad@amazon.com}}
\affil[3]{Princeton University, Princeton, NJ, USA \\
  \texttt{jrex@cs.princeton.edu}}
\affil[4]{Google Inc., Mountain View, CA, USA \\
  \texttt{kelvinzou@google.com}}

\authorrunning{M. Charikar, Y. Naamad, J. Rexford, and X. K. Zou} 
\Copyright{Moses Charikar, Yonatan Naamad, Jennifer Rexford, and Xuan Zou}


\begin{document}
\maketitle

\begin{abstract}
Modern networks run ``middleboxes'' that offer services ranging from network address translation and server load balancing to firewalls, encryption, and compression.  In an industry trend known as Network Functions Virtualization (NFV), these middleboxes run as virtual machines on any commodity server, and the switches steer traffic through the relevant chain of services.  Network administrators must decide how many middleboxes to run, where to place them, and how to direct traffic through them, based on the traffic load and the server and network capacity.  Rather than placing \emph{specific} kinds of middleboxes on each processing node, we argue that server virtualization allows each server node to host \emph{all} middlebox functions, and simply vary the fraction of resources devoted to each one.  This extra flexibility fundamentally changes the optimization problem the network administrators must solve to a new kind of multi-commodity flow problem, where the traffic flows consume bandwidth on the links as well as processing resources on the nodes.  We show that allocating resources to maximize the processed flow can be optimized exactly via a linear programming formulation, and to arbitrary accuracy via an efficient combinatorial algorithm.  Our experiments with real traffic and topologies show that a joint optimization of node and link resources leads to an efficient use of bandwidth and processing capacity.  We also study a class of design problems that decide where to provide node capacity to best process and route a given set of demands, and demonstrate both approximation algorithms and hardness results for these problems.
\end{abstract}

\renewcommand{\thefootnote}{\fnsymbol{footnote}}
\footnotetext[2]{ A much earlier version of this work is available online at \textcolor{blue}{\url{https://www.cs.princeton.edu/~jrex/papers/mopt14.pdf}}. The two works differ significantly in both content and presentation.}
  
\section{Introduction}
\label{sec:intro}
In addition to delivering data efficiently, modern networks often perform services on the traffic in flight to enhance security, privacy, or performance, or provide new features.  Network administrators often install ``middleboxes'' such as firewalls, network address translators, server load balancers, Web caches, video transcoders, and devices that compress or encrypt the traffic.  In fact, many networks have as many middleboxes as underlying routers or switches~\cite{APLOMB2012}.  Often a single connection must traverse multiple middleboxes, and different connections may go through different sequences of middleboxes.  For example, while Web traffic may go through a firewall followed by a server load balancer, video traffic may simply go through a transcoder.  To keep up with the traffic demands, an organization may run multiple instances of the same middlebox.  Deciding how many middleboxes to run, where to place them, and how to direct traffic through them, is a major challenge facing network administrators.

Until recently, each middlebox was a dedicated appliance, consisting of both software and hardware.  Administrators typically installed these appliances at critical locations that naturally see most of the traffic, such as the gateway connecting a campus or company to the rest of the Internet.  A network could easily have a long chain of these appliances at one location, forcing all connections to traverse every appliance---whether they need all of the services or not.  In addition, placing middleboxes only at the gateway does not serve the organization's many \emph{internal} connections, unless the internal traffic is routed circuitously through the gateway.
%

Over the last few years, middleboxes have become increasingly \emph{virtualized}, with the software service separate from the physical hardware---an industry trend called Network Functions Virtualization (NFV)~\cite{nfv,opnfv}.  The network can ``spin up'' (or down) virtual machines on any physical server, as needed.  This has led to a growing interest in good algorithms for optimizing the (i) \emph{allocation} of middleboxes over a pool of server resources, (ii) \emph{steering} of traffic through a suitable sequence of middleboxes based on a high-level policy, and (iii) \emph{routing} of the traffic between the servers over efficient network paths~\cite{SIMPLE2013,slick, jointoptvideo2015,NSDI2016,LiQian2016}.

Rather than solving these three optimization problems separately, we introduce---and solve---a joint optimization problem.  Since server resources are fungible, we argue that each processing node could subdivide its resources \emph{arbitrarily} across any of the middlebox functions, as needed.  That is, the \emph{allocation} problem is more naturally a question of what fraction of each node's computational (or memory) resources to allocate to each middlebox function. Similarly, each connection can have its middlebox processing performed on any node, or set of nodes, that have sufficient resources.  That is, the \emph{steering} problem is more naturally a question of deciding which nodes should devote a share of their processing resources to a particular portion of the traffic.  Hence, the joint optimization problem devolves to a new kind of \emph{routing} problem, where we compute paths based on both the bandwidth and processing requirements of the traffic between each source-sink pair.  That is, each flow from a source to a sink must be allocated both (i) a certain amount of bandwidth on \emph{every} link in its path and (ii) a \emph{total} amount of computational across all of the nodes on its path.

In our \emph{flow with in-network processing} problem, we have a flow demand with multiple sources and multiple sinks, and each flow requires a certain amount of processing. The required processing is proportional to the flow size and, without loss of generality, we assume one unit of flow requires one unit of processing. Each flow from a source to a sink is an aggregate flow of many connections, so the routing and processing for a flow are both divisible. In this model there are two types of constraints: edge capacity constraints and vertex capacity constraints, which represent link bandwidth and node processing capacity, respectively. A feasible flow pattern satisfies three conditions: (i) for each edge, the sum over all flows on that edge is bounded by that edge's capacity, (ii) for each node, the sum over all flows of in-network processing done at that node is bounded by the vertex capacity, and (iii) each flow must be allocated a total amount of node processing power equal to its size.

Although ignoring vertex capacity constraints reduces our class of problems to those of the standard multi-commodity flow variety, the introduction of these constraints yields a new class of problems that has not been studied before. This paper provides a systematic approach to this new class of network problems, applicable to both directed and undirected graphs. 

In \Cref{sec:problem}, we introduce the \textsc{processed flow routing} class of problems, in which we discuss how to optimize processed flow routed in a fixed network. 
Next, in \Cref{subsec:walk,subsec:edge} we present two linear programming-based algorithms to find a maximum feasible multi-commodity flow with the additional processing constraints. We show that, like standard multi-commodity flow, the program can be written in two different equivalent ways: either with an exponentially-sized walk-based LP or with a polynomially-sized edge-based LP. The proof of equivalence of equivalence of these two LPs requires a more careful argument than that for standard MCF. As an aside, we argue that this pair of LPs can also be adapted to optimize several other objective functions, such as those minimizing congestion. In \Cref{sec:mwu}, we then design an efficient multiplicative weight update (MWU) algorithm that finds approximately optimal solutions to our walk-based linear program far more quickly than one could with the edge-based program paired with an off-the-shelf LP-solver.

In \Cref{sec:networkdesign}, we consider the {\sc middlebox node purchase} class of problems, in which the goal is to optimally purchase processing power at various middleboxes. Prices for placing processing at the various nodes is given as part of the input, and may differ substantially from one location to the next. This class of problems comes has two natural variants: 

\begin{enumerate}
\item \minmp: given a set of flow demands, minimize cost while purchasing enough middlebox processing capacity so that all flow demands are simultaneously satisfiable (that is, jointly routable and steerable).
\item \maxmp: given a set of flow demands and a budget of $k$ dollars, spend at most $\$k$ on purchasing middlebox processing capacity while maximizing the fraction of the given demand that is simultaneously satisfiable.
\end{enumerate}

 Linear programs for both of these problems can be found in \Cref{dirminalg,dirmaxalg}. For \minmp, we show an $O(\log(n)/\delta^2)$ approximation for node costs and an associated multi-commodity flow that satisfies $(1-\delta)$ fraction of the demands and satisfies all edge capacities, where $n$ is the number of nodes. We show that in the directed case, the problem is hard to approximate better than a logarithmic factor, even if the demand requirements are relaxed. Additionally, we show that the undirected case is at least as hard to approximate as {\sc Vertex Cover}. 

We also prove approximation and hardness results for \maxmp. Although it's tempting to conjecture that the problem is an instance of {\sc Budgeted Submodular Maximization}, one can construct instances on both directed and undirected graphs where the amount of routable processed flow is \textit{not} submodular in the set of purchased nodes, so black-box submodular maximization techniques cannot be used here.
We show an $\Omega(1/\log(n))$ approximation for both problems, as well as a constant factor approximation algorithm for undirected instances with a single source-sink pair.
For the directed case, we show approximation hardness of $1-1/e$, and constant factor hardness for the undirected problem.

\section{Flow Routing with In-Network Processing}
\label{sec:problem}

\begin{figure}
	\centering
	\includegraphics[width=\linewidth]{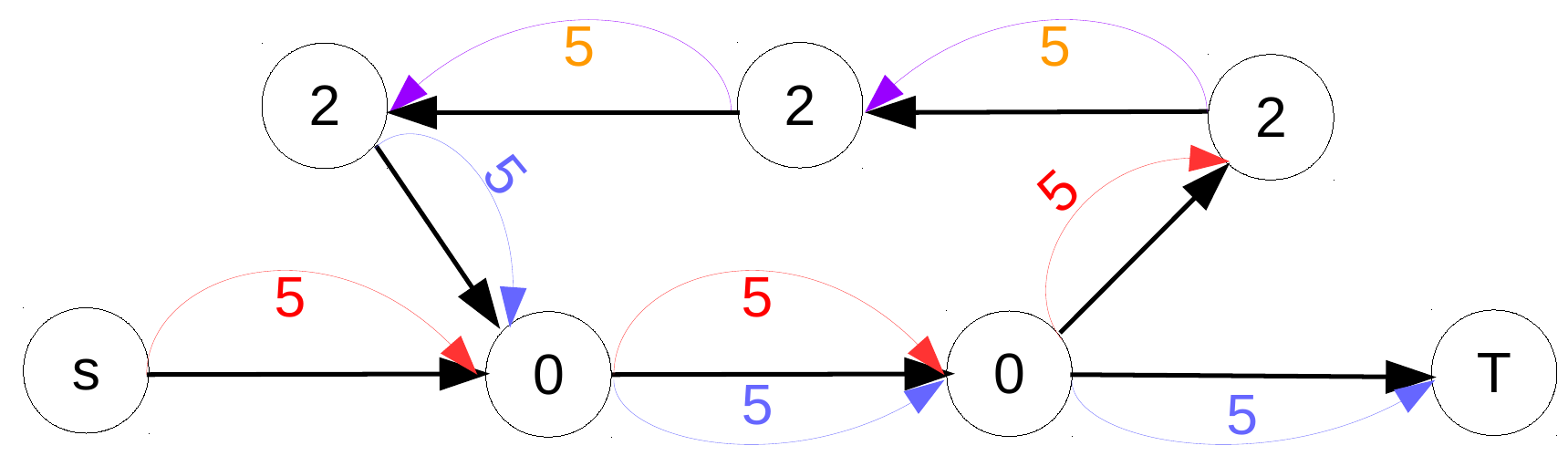}
	\caption{The edge capacity is 10 for all edges and the node capacities are denoted in each node. Here, we can send maximum flow size 5, by routing it along the red arcs, have it processed at the nodes at the top, and then sent to $T$ along the blue arcs. The capacity of the bottom middle edge forms the bottleneck here, as all flow must pass through it twice before reaching $T$.}
	\label{fig:loop}
\end{figure}

\subsection{Processed Flow Routing Problem}
\label{subsec:basicproblem}

Network Function Virtualization (NFV) allows each node to function as a general-purpose server that can run any in-network processing task, such as transcoding, compression, and encryption. Such servers can reside anywhere in the network, from the leaf nodes (as in the case of traditional servers) to intermediary nodes (such as top-of-rack and spine switches). 

Therefore, in our model, we treat all in-network processing as \emph{homogeneous}, meaning that every node with a sufficient quantity of available computational resources can be adapted to accomplish any processing task. In practice, this can be accomplished simply by spinning up a new virtual machine for that specific task as needed. We assume that all flows are both aggregate and sufficiently large that they can be treated as continuous quantities (and thus can be arbitrarily subdivided), and that the processing capacity of a given node can also be fractionally divided among a number of different flows. 

Each flow is initially generated at a source fully unprocessed. By the time it reaches its destination, it needs to go through and get processed by one or more intermediate processing nodes with available computational resources. We assume that each unit of flow requires one unit of processing, meaning that for any given flow $f$, the total processing workload done on $f$ by vertices along $f$'s flow path should equal the size $f$.

This problem can be modeled mathematically as follows. We are given a directed graph $G = (V,E)$ along with edge capacities $B : E \rightarrow \positivereal$, vertex capacities $C : V \rightarrow \nonnegativereal$, and a collection of flows $D = \{(s_1, t_1, k_1), (s_2, t_2, k_2), \cdots\} \subseteq V \times V \times \positivereal$. While the edge capacities are used in the same way as in a standard multi-commodity flow problem, we also require that each unit of flow undergo a total of one unit of processing at intermediate vertices. In particular, while edge capacities limit the \textit{total} amount of flow that may pass through an edge, vertex capacities only bottleneck the amount of processing that may be done at a given vertex, regardless of the total amount of flow that uses the vertex as an intermediate node. The goal is then either to route as much flow as possible, or to satisfy all flow demand subject to a congestion-minimization objective function. For concreteness, this paper focuses on maximizing the total amount of flow we can send between the source-destination pairs while satisfying edge and node capacity constraints. In practice we can also extend our results to other objective functions such as minimizing the weighted sum of congestion at edges and nodes.


\begin{table}
	\normalsize
	\begin{tabular}{ | c | l | }
		\hline
		\textbf{Variable} & \multicolumn{1}{c|}{\textbf{Description}} \\ \hline
		$V$ & set of nodes in a graph \\ \hline
		$E$ & set of edges in a graph\\ \hline	
		$B(e)$ & edge capacity for edge $e$\\ \hline
		$C(v)$ & node capacity for node $v$\\ \hline
		$D$ & the set of flow demands\\ \hline 
		$\delta^+(v)$ & the edges leaving vertex $v$\\ \hline  
		$\delta^-(v)$ & the edges entering vertex $v$ \\ \hline
		$P$ & the set of walks from sources to destinations \\ \hline
		$p_{i,\pi}^v$ & walk-based; the amount of flow $i$ from $s_i$ to $t_ i$ \\
		& exactly using  walk $\pi$ and processed at $v$\\ \hline
		$f_i(e)$ & edge-based; the amount of flow $i$ that \\
		& traverses $e$ on its way from $s_i$  to $t_i$ \\ \hline
		$w_i(e)$ & edge-based; the amount of unprocessed flow  $i$ \\
		& that traverses $e$ on its way from $s_i$ to $t_i$\\ \hline
		$p_i(v)$ & edge-based; the amount of processing done at \\ & node $v$ for the $i$th flow \\ \hline
	\end{tabular}
	\caption{Variables in the optimization solutions}
\end{table}


\subsection{A 2-Walk-based Solution}
\label{subsec:walk}

We now describe a \emph{\twalk-based} formulation of the problem. A \emph{\twalk} from $s$ to $t$ is a route between $s$ and $t$ that visits each vertex (and thus each edge) at most \emph{twice}. 

The approach we take is analogous to \textit{path-based} solutions for the traditional multicommodity flow (MCF) problem, with the key difference that, unlike paths, our \twalks may visit vertices and edges more than once.  Additionally, a \twalk may traverse the same undirected edge in both directions. 

To express the \textit{\twalk-based} linear program, we introduce one variable $p_{i,\pi}^v$ for each \{\twalk{}\}--vertex--demand triplet, representing the total amount of flow from $s_i,t_i$ exactly utilizing walk $\pi$ and processed at $v$. Note here the set $P$ of \twalks is an enumeration of all possible \twalks in the graph, which can be \emph{exponential} in size. The LP is then formulated as follows:

 \vspace{-.2in}
\begin{subequations}
	\begin{align*}
	&\textsc{maximize }& \Sigma_{i=1}^{|D|}\textstyle\sum_{\pi \in P}p_{i,\pi}
	\\  
	&\textsc{subject to} \\
	& p_{i,\pi} = \sum \limits_{v\in \pi} p_{i,\pi}^v& \forall i \in [|D|], \forall \pi \in P \\
	& \sum\limits_{i=1}^{|D|}\sum \limits_{\stackrel{\pi\in P}{\pi \ni e}} p_{i,\pi} \leq B(e) & \forall e \in E\\
	&\sum\limits_{i=1}^{|D|} \sum \limits_{\pi\in P} p_{i, \pi}^v \leq C(v) &\forall v \in V \\
	&p_{i,\pi}^v \geq 0 & \hspace{-.6in} \forall i\in[|D|],\forall \pi\in P,\forall v \in V
	\end{align*}
\end{subequations}
While the first constraint enforces that all flows are fully processed, the second and third constraints ensure that no edge or vertex is over-saturated. 

\section{An Edge-based Polynomially-Sized LP}
\label{sec:routing}

Although the \twalk-based solution exactly solves our MCF with in-network processing problem, the LP may be exponentially sized and thus even writing it down (let alone solving it) leaves us with an exponential worst-case running time. In \Cref{subsec:edge}, we present a polynomially-sized (and thus polytime-solvable) edge-based linear program for this problem. We then follow this up by a proof of correctness in \Cref{subsec:flowmax}.


\subsection{The Edge-Based Solution}
\label{subsec:edge}

A standard technique for solving the traditional MCF problem relies on constructing a polynomially-sized \emph{edge-based} LP whose set of feasible solutions equals that of an exponentially-sized path-based LP. Analogously, we establish a polynomial-sized edge-based LP corresponding to the \emph{\twalk-based} LP introduced in \Cref{subsec:walk}.

\begin{subequations}
	\footnotesize
	\begin{align}
	\hspace{-.2in}	&\textsc{maximize }&\sum\limits_{i=1}^{|D|}\sum \limits_{e\in \delta^+(s_i)} f_i(e)  \\ 
	&\textsc{Subject to}\\
	&\sum\limits_{e \in \delta^-(v)}  f_i(e)=  \sum\limits_{e \in \delta^+(v)} f_i(e)\hspace{-.25in}& \forall i \in [|D|], \forall v \in V \setminus \{s_i,t_i\}\\
	&p_i(v) = \sum\limits_{e \in \delta^-(v)} w_i(e) - \sum\limits_{e \in \delta^+(v)} w_i(e)  \hspace{-.25in} & \forall i \in [|D|],\forall v \in V \setminus \{s_i\}\\
	&\sum\limits_{i=1}^{[D]} f_i(e)\leq B(e)&\forall e \in E\\
	&\sum\limits_{i=1}^{|D|} p_i(v) \leq C(v)&\forall v \in V\\
	&w_i(e) \leq f_i(e)  &\forall i \in [D], \forall e \in E \\
	&w_i (e)= f_i (e)&\forall i \in [D],\forall e \in \delta^+(s_i) \\
	&w_i (e)= 0&\forall i \in [D], \forall e \in \delta^-(t_i) \\
	&w_i (e), p_i (v) \geq 0&\forall i \in [D],\forall e \in E 
	\end{align}
\end{subequations}

The LP constraints can be interpreted as follows. Constraint (2c) is a flow conservation constraint: at any non-terminal node of flow $i$, the amount of flow $i$ that enters the node should equal the amount that leaves it. Constraint (2d) is a \emph{processing conservation} constraint, ensuring that the total amount of flow (processed or unprocessed) going through a node remains unchanged, although the quantity of each might change if the node processes any of the flow. Constraints (2e) and (2f) ensure that we don't exceed the edge and node capacity. Constraint (2g) ensures that the amount of work yet to be done on a flow does not exceed the size of the flow itself, while (2h) and (2i) ensure that all flows leave the sources unprocessed and arrive to the destinations fully processed.


\subsection{Proof of Equivalence to the \Twalk-Based LP}
\label{subsec:flowmax}
While the construction of the edge-based LP is not particularly difficult, it is not obvious that the edge-based solution actually solves the problem in question. We need to prove the \emph{correctness} of the edge-based LP. A priori, solutions to the edge-based LP here may not be decomposable to a valid routing pattern at all. In this subsection, we provide an efficient algorithm converting feasible solutions to the edge-based LP into corresponding solutions to the \twalk-based program, proving both that the edge-based LP is correct and that the actual flow paths can be recovered in polynomial time as well. 
We summarize this result in the following theorem.

\begin{theorem}
\label{thm:flowmax}
The edge-based formulation provides a polynomial-sized linear program solving the Maximum Processed Flow problem. Further, the full routing pattern can be extracted from the LP solution by decomposing it into its composing $s_i, t_i$ \twalks in $O(|V| \cdot |E| \cdot |D| \cdot  \log |V|)$ time.
\end{theorem}

Notably, as the reduction maps the set of feasible solutions to the edge-based LP to equivalent feasible solutions of the \twalk-based LP, the same technique can also be used to show the equivalence of the two corresponding programs when the objective function is changed to optimize some other linear quantity, such as the amount of congestion. 

The first part of the proof involves showing how to construct a solution to the flow-based LP when there is exactly one $s_i, t_i$ pair. Extracting the corresponding flow paths and iterating this procedure for each demand pair eventually extracts all $s_i,t_i$ flows, giving us a solution to the multicommodity problem.

The flow extraction argument proceeds in two steps. First, we simplify the solution by removing extraneous loops that do not affect the optimal solution. Next, we show that the existence of any residual flow in the graph (i.e., the existence of some strictly positive $f_i(e)$) implies that there exists at least one valid \twalk we can efficiently extract while maintaining feasibility of all constraints for the updated residual graph. As we show, a linear number of extractions suffices to remove all flow from the solution.

\subsubsection{Removing extraneous loops}
Suppose we are given a nonempty solution to the \emph{edge-based} LP for an instance with graph $G(V,E)$. We focus on some (arbitrarily chosen) commodity $i$ with positive flow in this solution, and drop subscripts to let $f(e)$, $w(e)$, and $p(v)$ denote $f_i(e)$, $w_i(e)$, and $p_i(v)$, respectively. To assist with our exposition, we restrict our attention to the subgraph $G'$ which excludes all edges for which $f(e) = 0$. For each edge $e$ in this subgraph, we also associate two new variables, $f^1(e)$ and $f^2(e)$ denoting the amount of unprocessed and processed flow passing through this edge, respectively. Thus, by definition, $f^1(e) = w(e)$ and $f^2(e) = f(e) - w(e)$.

As in solutions to the edge-based linear program for the standard multicommodity flow problem, solutions to our edge-based LP may introduce closed loops (that is, directed cycles along which a positive amount of flow is routed). In traditional MCF, such loops are easily shown to be \emph{non-essential}, and can be easily removed from a feasible solution without affecting its correctness. As illustrated in \Cref{fig:loop}, such loops may actually be critical in solutions to our variant, and handling such cases takes additional care. Thus, instead of arguing that cycles can be removed (so that the flows form a set of paths), we show how to ensure that no vertex may be visited more than twice (and thus the flows form a set of \twalks). In particular, we show how to cancel out all cycles along which each edge contains $f^1$ flow, as well as all cycles along which each edge carries $f^2$ flow.

\begin{lemma} Any closed loop for which every edge contains $f^1$ (resp. $f^2$) flow can be removed without affecting the total $(s,t)$ flow. 
\label{lem:flowcancel}
\end{lemma} 

\begin{proof} This argument proceeds similarly to the flow cancellation arguments in the traditional MCF setting. For any loop $l$ containing a positive amount of $f_1$ flow, reducing both $f(e)$ and $w(e)$ on the constituent edges by $\min_{e' \in l} f^1(e')$ ensures that all constraints in the LP remain satisfied. For loops containing a positive amount of $f_2$ flow, similarly reducing just $f(e)$ suffices.
\end{proof}

\subsubsection{Extracting \twalks} Suppose extraneous loops have been removed using the process described in \Cref{lem:flowcancel}. Define $\rho_{e} = \frac{w(e)}{f(e)} = \frac{ f^1(e)} {f^1(e)+f^2(e)}$. By \Cref{lem:flowcancel}, every cycle with a positive $f(e)$ on each edge contains at least one edge with $\rho = 1$ and another with $\rho = 0$. We now repeat the following until all flow is removed from the graph. Select a vertex $v$ that is allocated processing (i.e., $p(v) > 0$), and run a backwards traversal from $v$, at each step selecting the incoming edge with the largest fraction of unprocessed flow (i.e., maximizing $\rho(e)$) until we reach $s$. Similarly, run a forward traversal from $v$ to $t$ along edges minimizing $\rho$.  This route will be our ``flow-\twalk''. The amount of flow routable along this flow-\twalk is the minimum of three quantities: (1) the smallest amount of unprocessed flow sent on each edge of the $s \leadsto v$ path, (2) the smallest amount of processed flow sent along each edge of the $v \leadsto t$ path, and (3) the amount of processing still to be done at $v$ (i.e., $p(v)$). We then extract this flow-\twalk from the solution by decreasing each LP variable accordingly. Complete pseudo-code for this algorithm is given in \Cref{cap:algorithm}. 

\begin{algorithm}
	\SetAlgoLined\DontPrintSemicolon
	\KwData{$G'(V, E)$, $ w(e)$, $f(e)$ for $\forall e \in E$ and $p(v)$ for $\forall v \in V$ }
	\KwResult{$f(\pi)$, $p(\pi,v)$ with $v\in\pi$ }
	\SetKwFunction{twowalk}{TwoWalkConstruction}  
	\SetKwFunction{flow}{FlowPlacement}  
	\SetKwProg{myalg}{Algorithm}{}{}
	\myalg{\twowalk{$s$, $t$, $v$}}  {
		\emph{//Construct \twalk from $s\rightarrow v$ and $v\rightarrow t$}\\
		From $v$, run a backward traversal, each time picking an incoming edge $e$ maximizing $\rho(e) = w(e)/f(e)$\\
		From $v$, run a forward traversal, each time picking an outgoing edge minimizing $\rho(e) = w(e)/f(e). $\\
		\KwRet $\pi$\\}{}
	\setcounter{AlgoLine}{0}
	\SetKwProg{myproc}{Algorithm}{}{}
	\myproc{\flow{$s$, $t$}}{
		\BlankLine
		\While{there exists a $v$ with $p(v) > 0$}
		{
			$\,\,\,\pi \leftarrow \twowalk{s, t, v}$\\
			$f' \leftarrow \min_{e \in \pi, e\textrm{ precedes } v} f^1(e)$ \\
			$f'' \leftarrow  \min_{e \in \pi, e\textrm{ succeeds } v} f^2(e)$ \\
			$p_\pi^v \leftarrow \min\{f', f'', p(v)\}$\\
			\For{$u\in \pi$ and $u\not=v$ }
			{
				$p_\pi^u=0$\\
			}
			$\,\,\,C(v) \leftarrow C(v) - p_\pi^v$\\
			$p(v) \leftarrow p(v) - p_\pi^v$\\
			\For{$e \in \pi$}{
				$\,\,\,f(e)\leftarrow f(e)-p_\pi^v$\\
				$B(e) \leftarrow B(e)-p_\pi^v$
			}
		}
	}
	\caption{\Twalk Decomposition}\label{cap:algorithm}
\end{algorithm}

\begin{lemma} [\Twalk Extraction] \Cref{cap:algorithm} can always generate a \twalk with non-zero flow from source to sink if there exists any $v$ where $p(v) >0$. Further, the number of iterations needed of \Cref{cap:algorithm} is bounded by $O(|E|)$, each of which can be made to take $O(|V| \log |V|)$ time. Thus, the total running time is $O(|E|\cdot|V| \log |V|)$
\end{lemma}

\begin{proof}
The removal of extraneous cycles guarantees that no \twalk can visit the same vertex more than twice. Now suppose that a vertex $v$ has $p(v) > 0$. By constraint (2d), the $f^1$ flow on some incoming edge and the $f^2$ flow on some outgoing edge must both be positive. By a combination of constraints (2c), (2d), and (2j), the reverse traversal from $v$ to $s$ must succeed: it cannot get ``stuck'' at a vertex $u$ with no in-edge with positive $f^1$ flow. Similarly, the forward traversal from $v$ to $t$ must find a path with positive $f^2$ on each edge. Subtracting the minimum of all of the reverse path's observed $f^1$ values, all of the forward path's observed $f^2$ values, and $p(v)$ from each of those variables ensures that all variables remain nonnegative. Further, as this operation is monotone and it decreases one of the variables to $0$, repeating this must remove all flow from the graph in at most $|V| + 2|E| = O(|E|)$ iterations. By initially constructing a priority queue for each vertex on the $f^1$, $f^2$, and $\rho$ values of its neighboring edges and updating them accordingly, the forward and backward traversals can be found in $|V| \log |V|$ time, each.
\end{proof}

We can generalize the above approach to the multicommodity problem by treating each of the commodities independently. Namely, sequentially applying the above algorithm to remove flow \twalks for each of the $|D|$ demand pair gives us a solution to the multicommodity problem without violating any of the LP constraints. Thus, we get the $O(|V|\cdot|E|\cdot|D|\cdot\log |V|)$ running time promised in the statement of \Cref{thm:flowmax}.

\section{A Multiplicative Weights Algorithm}
\label{sec:mwu}
Solving LPs as large as those described in \Cref{subsec:edge} can be expensive. Although the LP solver \emph{CoinLP} takes an average of only 2.7 seconds on traces from the 12-node Abilene network, it takes over a minute when the size of the network grows to include a still-modest 35 nodes\footnote{All computations done on a 3.3GHz Intel i5-2500K processor.}. For even larger networks, we need an entirely different approach. Previously, the \emph{multiplicative weight update} (MWU) method has been widely used to efficiently approximate the optimal solution to traditional multicommodity flow LPs~\cite{mwu,pst}. We show that, with added effort, this framework can be applied even in the presence of node processing constraints, giving $(1-\epsilon)$ approximation to the problem in time $O(|D|\cdot |E|\cdot (|E| + |V| \log |V|) \cdot \log^2 |V| / \epsilon^2)$.

We first briefly overview the MWU method in \Cref{subsec:mwuoverview}. Next, we describe how to apply the MWU method to our model including processing vertices. The proof of correctness is given in \Cref{subsec:mwuproof};

\subsection{Multiplicative Weight Update for Traditional MCF}
\label{subsec:mwuoverview}

In the traditional multiplicative weights algorithm for multicommodity flow, there an ``expert'' is assigned to each edge, each of which is initially assigned a sufficiently small weight. The algorithm then iteratively finds $s_i, t_i$ walks minimizing the sum of weighted utilization of their edges and adds together scaled down versions of these paths to eventually construct a solution.
When a path is chosen, all experts corresponding to edges along the path have their weight increased by a multiplicative factor, making it less likely that we repeat our selection of the edges. 
This process is repeated until some expert's weight surpasses the value $1$, corresponding to a fully utilized edge. When this happens, all paths are scaled down by the weight of the largest expert to ensure that no capacities are exceeded. One then shows that the final result is within a $(1-\epsilon)$ factor of the maximum multicommodity flow. 

\subsection{Formulation and Analysis}
\label{subsec:mwuformulation}
Although we derive the same $(1-\epsilon)$ approximation factor for our problem, the analysis of our multiplicative weights algorithm is quite different from that of traditional multicommodity flow. Intuitively, this is because vertex capacities are inherently very different from edge capacities: while a flow \twalk reduces the remaining capacity on \emph{all} edges it traverses, it only reduces the capacity for \emph{one} of its vertices. Thus, we set up a different update condition, as well as a different method for picking the best flow \twalks for each round.

\subsubsection{Setup} 
For each edge $e$, we have a constraint $\sum_{\pi} p_{\pi} \leq B(e)$, where $p_{\pi}$ is the amount of flow sent on \twalk $\pi$. For each vertex, the corresponding constraint is $\sum_{\pi} p_{\pi}^v \leq C(v)$, where $p_{\pi}^v$ is the amount of flow on \twalk $\pi$ that is processed at $v$. For each of these two sets of constraints, we associate one expert, (which we call $\hat{e}$ and $\hat{v}$), whose weights are denoted by $w_{\hat{e}}$ and $q_{\hat{v}}$, respectively. 

Consider a feasible solution to the \twalk-based LP. The feasible solution consists of variables of the form $p_{\pi}$ and $p_{\pi}^v$. In this section, we abuse notation and let the variable $p$ denote a feasible solution to the LP, at which point $p_{\pi}$ and $p_{\pi}^v$ become bound variables for each $\pi$ and $v$ (that is, $p$ can be thought of as a dictionary containing the aforementioned set of variables). Further, define $A(p)$ as the objective function value of $p$, i.e. $A(p) = \sum_{v \in V} \sum_{\pi \in P} p_{\pi}^v$.

For an expert $\hat{e}$ and feasible solution $p$, define the gain $M(\hat{e},p)$ by \( M(\hat{e},p) = \frac{1}{B(e)}\sum_{\pi\ni e} p_{\pi}., \) This can be thought of as the fraction of $e$'s capacity actually utilized by the feasible solution. For each expert $\hat{v}$, we define the gain $M(\hat{v},p)$ by \( M(\hat{v},p) = \frac{1}{C(v)}\sum_{\pi\ni v} p_{\pi}^v, \)
which corresponds to the fractional utilization of $v$'s processing capacity.

Let $\mathcal{D}$ be the probability distribution over experts in which the probability of choosing a given expert is proportional to its weight. For a fixed $p$, the expected gain of a random variable sampled from $\mathcal{D}$ is

$$M(\mathcal{D}, p) = \frac{\sum_{e} w_{\hat{e}} M({\hat{e}},p) + \sum_v q_{\hat{v}} M({\hat{v}},p)}{\sum_{e}w_{\hat{e}} + \sum_v q_{\hat{v}}} $$

We first make two observations: 

{\bf Observation 1: } For any feasible solution $p$, $0\le M(\mathcal{D}, p) \le 1$. This is because $M(\hat{e},p)\le 1$ and $M(\hat{v},p)\le 1$ for all $e$ and $v$. 

{\bf Observation 2: } For any feasible solution $p$ and weights $w,q$, if $\pi^* = \textrm{argmin}_{\pi} \left(\sum_{e\in \pi}w_{\hat{e}}/B(e) + \min_{\hat{v}\in \pi} q_{\hat{v}}/C(v)\right)$, then $$M(\mathcal{D},p)\ge \frac{ A(p)\left(\sum_{\hat{e}\in \pi^*}w_{\hat{e}}/B(e) + \min_{\hat{v}\in \pi^*} q_{\hat{v}}/C(v)\right)}{\sum_{e}w_{\hat{e}} + \sum_v q_{\hat{v}}}$$

This is due to the fact that:  

\begin{align*}
\allowdisplaybreaks
M(\mathcal{D},p) 
&= \frac{\sum_{e} w_{\hat{e}} M(\hat{e},p) + \sum_v q_{\hat{v}}M(\hat{v},p)}{\sum_{e}w_{\hat{e}} + \sum_v q_{\hat{v}}} \\
& = \frac{\sum_{\pi} \left(p_{\pi} (\sum_{e}w_{\hat{e}}/B(e)) + \sum_{v\in \pi }p^v_{\pi}q_{\hat{v}}/C(v)\right) }{\sum_{e}w_{\hat{e}} + \sum_v q_{\hat{v}}} \\
& \ge \frac{\sum_{\pi} \left(p_{\pi} (\sum_{e}w_{\hat{e}}/B(e) + \min_{v\in \pi} q_{\hat{v}}/C(v)\right) }{\sum_{e}w_{\hat{e}} + \sum_v q_{\hat{v}}} \\
& \ge \frac{\sum_{\pi} p_{\pi} \min_{\pi}(\sum_{e\in \pi }w_{\hat{e}}/B(e)+ \min_{v\in \pi} q_{\hat{v}}/C(v))}{\sum_{e}w_{\hat{e}} + \sum_v q_{\hat{v}}} \\
& \ge \frac{ A(p)(\sum_{e\in \pi^*}w_{\hat{e}}/B(e)+ \min_{v\in \pi^*} q_{\hat{v}}/C(v))}{\sum_{e}w_{\hat{e}} + \sum_v q_{\hat{v}}}  
\end{align*}

Where $\pi^*$ is the path minimizing the $\argmin$ in the statement of the observation. Thus, in each round, we aim to find the $\pi^*$ minimizing this value. Conditioned on us being able to do so, the rest of the MWU algorithm proceeds as follows:

\begin{enumerate}
	\item We initialize all expert weights $\{w_{\hat{e}}\}$ and $\{q_{\hat{v}}\}$ to $1/\delta$, where $\delta = (1+\epsilon) ((1+\epsilon)\cdot|E|)^{-1/\epsilon}$. This choice of $\delta$ will be justified in the analysis of \Cref{subsec:mwuproof}.
	\item  At each step $t$, given weights $w_e^t$ and $q_v^t$ on the experts, we pick the flow-\twalk $p^t$ minimizing the quantity $\sum_{e\in \pi } \frac{w_{\hat{e}}}{B(e)} + \min_{v\in \pi }\frac{q_{\hat{v}}}{C(v)}$. An efficient algorithm for finding such a \twalk is given in \Cref{subsec:shortest-path}.
	\item Given the \twalk $p^t$ chosen in the previous step, we treat this as a feasible solution to the instance, giving expert $\hat{j}$ a gain of $M(\hat{j},p^t)$. Consequently, the weight $w_{\hat{e}}$ or $q_{\hat{v}}$ of each expert $j$ is increased by a multiplicative factor of $M(\hat{j}, p^t)$.
	\item The algorithm stops when one of the weights $w_{\hat{e}}$ or $q_{\hat{v}}$ is larger than 1. Once the algorithm terminates, we scale down the flow $p^t$ computed at each round by a factor of $\log_{1+\epsilon}\frac{1+\epsilon}{\delta} = 1 - \frac{\ln \delta}{\ln 1 + \epsilon}$, and return the set of all flow-\twalks $p^t$.
\end{enumerate}


\subsubsection{Computing the Optimal Path}
\label{subsec:shortest-path}

To compute the \twalk $\pi^t$ with minimum cost, we use a dynamic programming algorithm reminiscent of Dijkstra's shortest path algorithm. Given a graph $G(V,E)$, with weights $w(e)$ on edges, weights $n(v)$ on nodes, and some source-sink pair $s,t$, we are interested in computing the following quantity 

\begin{eqnarray}
opt(s,t) := \argmin_{\pi = (s, \cdots, t), v \in \pi} cost(\pi,v)
\end{eqnarray}

where $cost(\pi,v)$ is defined as 

\[ cost(\pi,v ): =\left(\sum_{e\in \pi} w(e) + n(v)\right) \]

We compute $opt(s,t)$ in two stages. First, for every $v$, we upper bound the value of $opt(s,v)$ by $n(v)$ plus the shortest distance from $s$ to $v$. Afterwards, we use dynamic programming to iteratively decrease these upper bounds. Full details are given in \Cref{alg:optpath}.

\begin{algorithm}\label{alg:generalized_shortest_path}
	\begin{algorithmic}[]
		\REQUIRE Graph $G = (V,E)$ with edge weights $w(e)$, node weights $n(v)$, and a designated source $s$.
		\RETURN $r(v) =  opt(s,v)$ for every $v\in V$. 
		\STATE Use Dijsktra's algorithm to compute the shortest path $d(v)$ between $s$ and $v$.
		\STATE Initialize $r(v)\leftarrow d(v) + n(v)$ for all $v\in V$. $S\leftarrow \{s\}$. 
		\WHILE{$S\neq V$}
		\STATE Let $u^* = \argmin_{v \in V \setminus S} r(v)$. Add $u^*$ to $S$.
		\STATE For all neighbors $z$ of $u^*$ that are not already in $S$, let $r(z) \leftarrow \min \{r(u^*)+ w(u,z), r(z)\}$ 
		\ENDWHILE
	\end{algorithmic}
	\caption{Optimal \Twalk Algorithm}
	\label{alg:optpath}
\end{algorithm}

\subsubsection{Update}
Suppose now that the \twalk $\pi$ with smallest cost has been computed. One of two things may bottleneck the amount of processed flow that can be sent along $\pi$: either the edge capacity of some edge $e$, or the processing capacity of some vertex $v$. We consider the two cases separately

If the bottleneck is edge-based, i.e. $\sum_{v\in \pi^t} C(v) \geq \min_{e\in \pi^t} B(e)$, then let $e^t = \argmin_{e\in \pi^t} B(e)$, and let the chosen flow \twalk $p^t$ be the one satisfying
\[
p^{t,v}_{\pi} = \left\{ \begin{array}{cc} 
\frac{C(v)}{\sum_{v\in \pi^t} C(v)} \cdot B_{e^t} & \textrm{ if } \pi  = \pi^t, v\in \pi^t\\
0 & \textrm{ otherwise }
\end{array}\right.
\]
On the other hand, if $\sum_{v\in \pi^t} C(v) < \min_{e\in \pi^t} B(e)$, select $p^t$ to satisfy
\[
p^{t,v}_{\pi} = \left\{ \begin{array}{cc} 
C(v) & \textrm{ if } \pi  = \pi^t, v\in \pi^t \\
0 & \textrm{ otherwise }
\end{array}\right.
\]


\subsection{Proof of the $(1-\epsilon)$ Approximation}
\label{subsec:mwuproof}

Let $T$ be the number of rounds taken until we hit the stopping criterion, and let $\bar{p} = \sum_{t=1}^T p^t$ be the total amount of flow selected after $T$ rounds. 
By the guarantee of the multiplicative update method (Theorem 2.5 in ~\cite{mwu}), we have that for any $e$ and any $v$

$$\sum_{t=1}^T M(\mathcal{D}^t,p^t) \ge \frac{\ln (1+\epsilon)}{\epsilon} M(\hat{e},\bar{p}) -\frac{\ln m}{\epsilon} $$

$$\sum_{t=1}^T M(\mathcal{D}^t,p^t) \ge \frac{\ln (1+\epsilon)}{\epsilon} M(\hat{v},\bar{p}) -\frac{\ln m}{\epsilon} $$

Since at time $T$, $w_{\hat{e}}^T = w^0_{\hat{e}}(1+\epsilon)^{M(\hat{e},\bar{p})}$, and $q_{\hat{v}}^T = w^0_v(1+\epsilon)^{M(\hat{v},\bar{p})}$, and the stopping rule ensures that at there exists $e$ or $v$ such that $w_{\hat{e}}^T\ge 1$ or $q_{\hat{v}}^T\ge 1$, we have that either there exists an $e$ such that $M(\hat{e},\bar{p})\ge \frac{\ln 1/\delta}{\ln (1+\epsilon)}$ or there exists $v$ such that $M(\hat{v},\bar{p})\ge \frac{\ln 1/\delta}{\ln (1+\epsilon)}$. Therefore, by the guarantee of the MWU method, we have that
$$\sum_{t=1}^T M(\mathcal{D}^t,p^t) \ge  \frac{\ln 1/\delta}{\epsilon} - \frac{\ln m}{\epsilon}$$

We now attempt to bound the left-hand-side of the preceding inequality. Note that $$M(\mathcal{D}^t,p^t) = \frac{\sum_{e} w^t_{\hat{e}} M(\hat{e}^t,p^t) + \sum_v q^t_{\hat{v}}M(\hat{v}^t,p^t)}{\sum_{e}w^t_{\hat{e}} + \sum_v q^t_{\hat{v}}}$$ $$ = \frac{A(p^t)\cdot (\sum_{e\in \pi^t}w^t_{\hat{e}}/B(e) + \min_{v\in \pi^t} q^t_{\hat{v}}/C(v))}{\sum_{e}w^t_{\hat{e}} + \sum_v q^t_{\hat{v}}} $$

By the definition of $\pi^t$ and Observation 2, we have 

\begin{align*}
M(\mathcal{D}^t,p^t) & = \frac{A(p^t)(\sum_{e\in \pi^t}w^t_{\hat{e}}/B(e)+ \min_{v\in \pi^t} q^t_{\hat{v}}/C(v))} {\sum_{e}w^t_{\hat{e}} + \sum_v q^t_{\hat{v}}} \\
& \le A(p^t) / A(p^{opt})
\end{align*}

Combining these inequalities, we get that 

$$A(\bar{p}) / A(p^{opt}) \ge\sum_{t=1}^T M(\mathcal{D}^t,p^t) \ge \frac{\ln 1/(|E|\cdot\delta)}{\epsilon}$$

Fixing any edge $e$, its expert's initial weight is $1/\delta$ and its expert's final weight is at most $1+\epsilon$. Thus, $\bar{p}$ passes at most $B(e)\log_{1+\epsilon}((1+\epsilon)/\delta)$ flow through it. Similarly, for each $v$, at most $C(v)\log_{1+\epsilon}((1+\epsilon)/\delta)$ units of processing are assigned to it. In other words, scaling down all $p^t$ flows by $\log_{1+\epsilon}(1+\epsilon)/\delta $ will result in a feasible flow. Letting $p' = \bar{p}/\log_{1+\epsilon}\frac{1+\epsilon}{\delta}$, we get 

\begin{align*}
 A(p') / A(p^{opt}) & \geq A(\bar{p})/ \left(A(p^{opt}) \log_{1+\epsilon}\frac{1+\epsilon}{\delta}\right) \\
 & \geq \frac{\ln (1/(|E|\cdot\delta))}{\epsilon}/\log_{1+\epsilon}\frac{1+\epsilon}{\delta} \end{align*}

Taking $\delta = (1+\epsilon) ((1+\epsilon)m)^{-1/\epsilon}$, we have that 
$$ \frac{A(p') }{ A(p^{opt})} \ge (1-\epsilon),$$

giving the promised approximation factor. 

Note that in each iteration, we either increase the weight of one $w_{\hat{e}}$ by a factor of $(1+\epsilon)$, or increase all of the $q_{\hat{v}}$'s on a path $\pi^t$ by a factor of $(1+\epsilon)$. Since each $w_{\hat{e}}$ and each $q_{\hat{v}}$ can only be increased by such a factor at most $\frac{\ln 1/(|E|\cdot \delta)}{\epsilon}$ times before its weight exceeds $1$, the total running time $T$ is bounded by $(|V| + |E|)\frac{\ln 1/(|E| \cdot \delta)}{\epsilon} \cdot T_{sp} = O(|E|\log |V|/\epsilon^2 \cdot T_{sp})$, where $T_{sp}$ is the time it takes \Cref{alg:optpath} to compute the optimal path for each the $|D|$ flows. As a single flow takes time $O(|E| + |V|\log V)$ using Fibonacci heaps, we can compute the \twalk for each of the flows in time $O(|D| \cdot(|E| + |V| \log |V|))$. Thus, the total running time is $O(|D| \cdot |E| \cdot (|E| + |V| \log |V|) \cdot \log^2 |V| / \epsilon^2)$.

\section{Evaluations}
We ran several experiments to address the following: (i) how well the LP fare against ``naive'' algorithms, and (ii) what is the runtime for an edge-based LP solution?

\subsection{Throughput Improvement}
To determine how well the LP fares against simple approaches, we compare it to a ``naive'' algorithm that first routes flow without vertex capacities in mind, and then processes as much flow as possible on the flow paths it initially routed. This is a variant of the \textit{path-selection} approach used in \cite{heorhiadi2015accelerating}. While there are simple examples where the naive algorithm performs extremely poorly in theory, we seek to study the performance in practice.

\begin{figure} [ht]
	\centering
	\begin{subfigure}{.47\textwidth}
		\hspace{-1in}
		\,\,\,\,\,\,\,\,\,\,\,\,\,\,\,\,\,\,\,\includegraphics[width=1.4\linewidth]{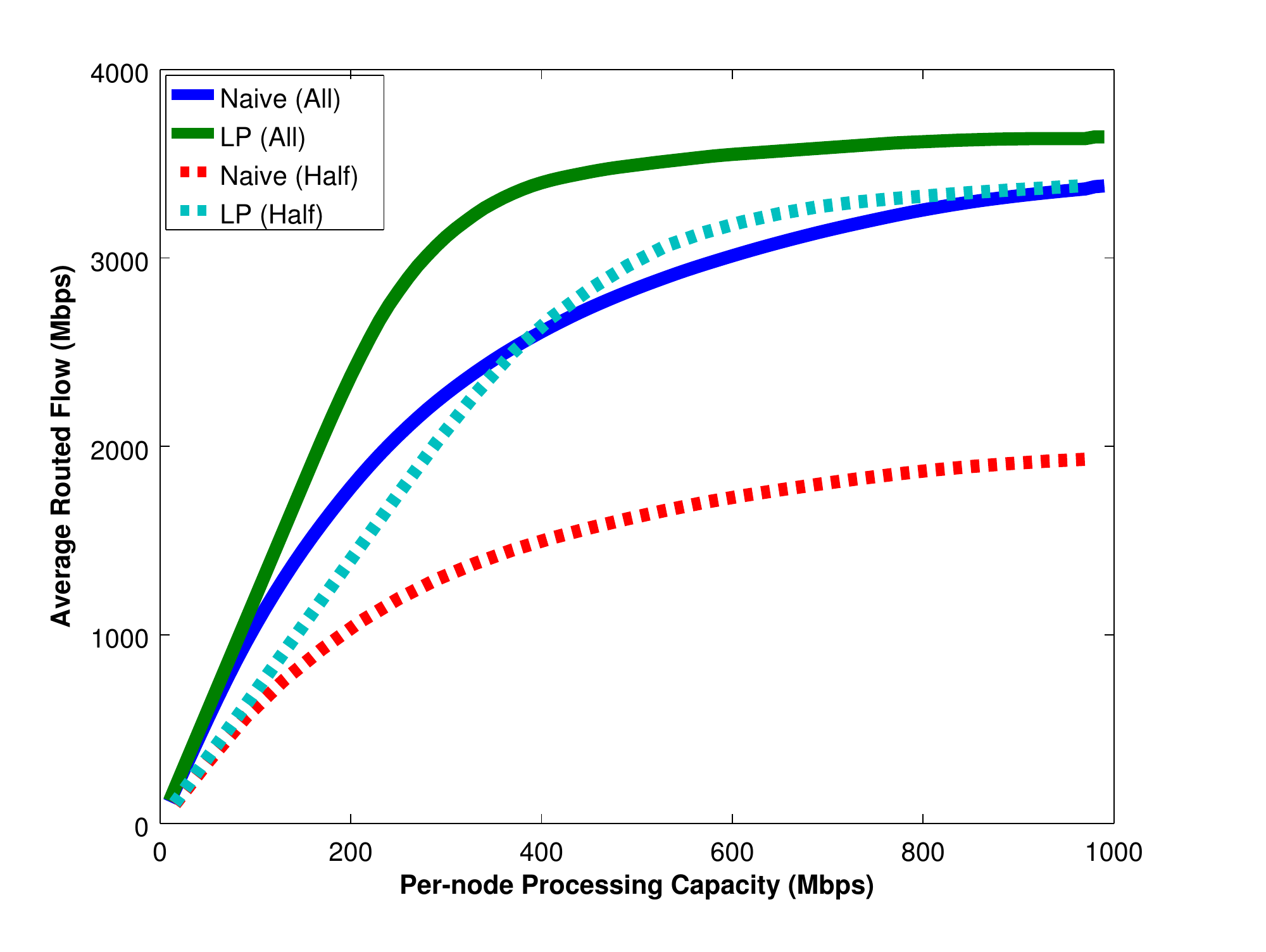}
		\caption{Amount of flow that the two algorithms could process given various node processing capacities.}
		\label{fig:expabsolute}
	\end{subfigure} \hspace{.2in}
	\begin{subfigure}{.47\textwidth}
		\hspace{-1in}
		\,\,\,\,\,\,\,\,\,\,\,\,\,\,\,\,\,\,\,\,\,\,\,\,\,\,\,\,\,
		\includegraphics[width=1.4\linewidth]{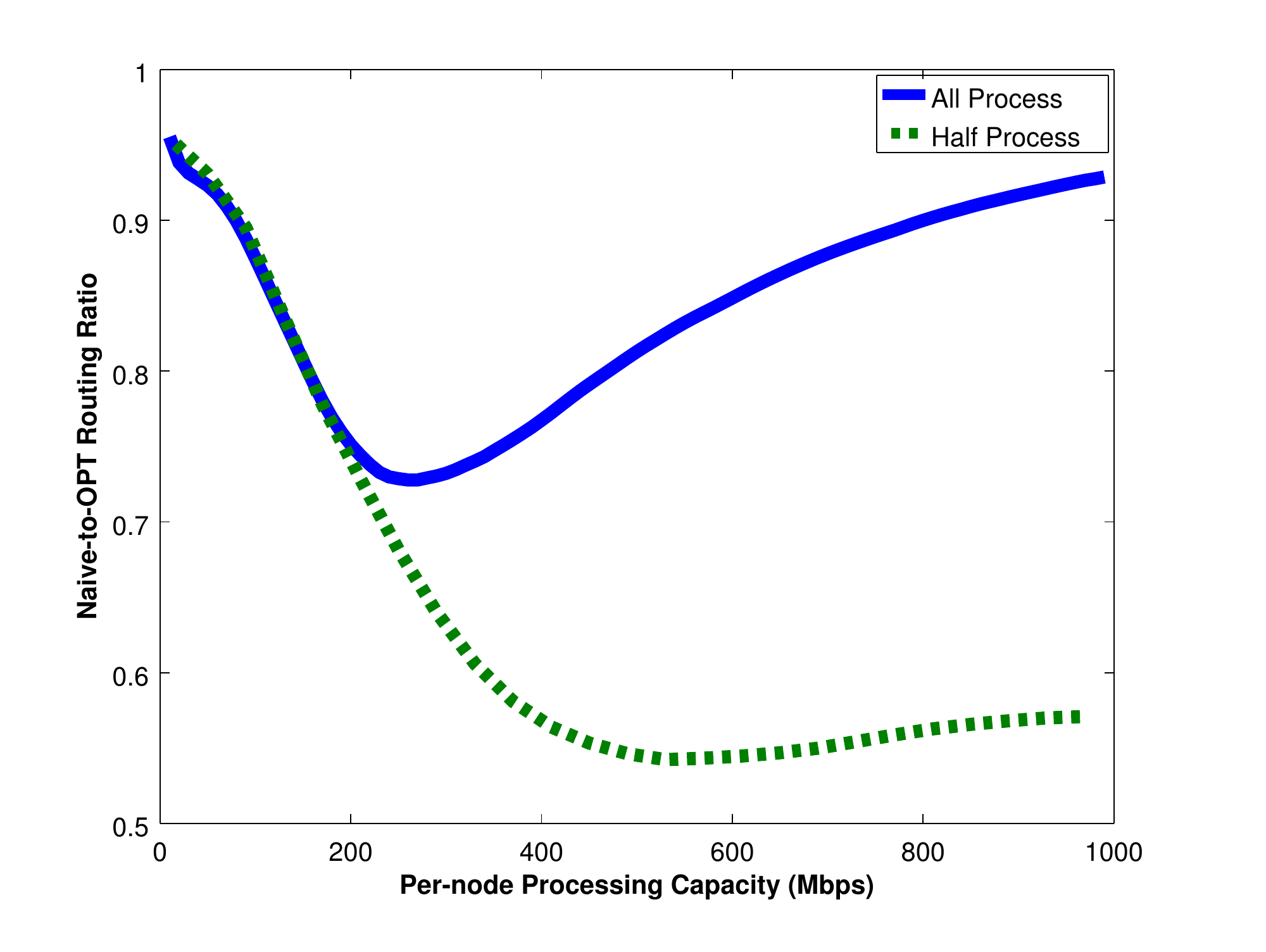}
		\caption{The ratio of the demanded flow processable by the naive algorithm to that processable by the LP, plotted at various processing capacities.}
		\label{fig:expratio}
	\end{subfigure}
	\caption{Experimental results showing how much of the demand both the naive and the (optimal) LP-based algorithm could successfully route and process given the Abilene traffic matrices.}
	\label{fig:experiments}
\end{figure}

We ran both algorithms on 150 randomly sampled traffic matrices provided by the TOTEM project \cite{uhlig2006providing} for the Abilene network in 2004.  As these datasets don't include vertex processing capacities, we compared the two algorithms on a wide range of values, with processing capacities assigned according to one of two distributions: either they all have the same capacity (the \textit{all} case) or exactly half of them have the prescribed capacity and the other half have zero (the \textit{half} case). The results are diagrammed in \Cref{fig:experiments}.

Experimental analysis shows that while the LP and the ``naive'' algorithm fare similarly when the network is low on processing capacity and thus node-throttled, or, in the \textit{all} case, high on node capacity and thus bottlenecked by the link capacities and the demand itself, the LP has a distinct advantage in between the two extremes when either resource could become the bottleneck when the flows are not routed efficiently. Additionally, the experiments show that the naive algorithm suffers when processing  is not uniformly distributed among the nodes even in the high-capacity case, as many of the initial flow paths might go entirely through nodes without any processing capacity and thus fail to get processed. Our experiments show that using the exact algorithm gives an improvement of up to $30\%$ over the naive approach if processing power is available at all nodes, and up to $80\%$ if the processing power is only placed at half of the nodes. 

\subsection{Runtime Analysis}
Although the edge-based LP provides a polynomial running-time guarantee,  it may still be too slow in practice on large graphs. To study the performance of linear programming, we also run the LP solver over a number of topologies acquired from SNDLib \cite{orlowski2010sndlib}. For each of the topologies in \Cref{tab:runtime}, enough processing capacity was evenly distributed among a random sample of half of all nodes so that the total processing capacity equals half of the total demand. As the table shows, the time to solve the LP grows quickly with the input size, and the solver ends up taking nearly two minutes for instances with 37 nodes. 

\begin{table}
\centering
\begin{tabular}{|l|c|c|c|}
     \hline
     \textbf{Network}& $|V|$ & $|E|$ & \textbf{Time (sec)}  \\ \hline
     \textbf{abilene} & 12 & 15 & 1.91 \\ \hline
     \textbf{dfn-bwin} & 10 & 45 & 3.08 \\ \hline
     \textbf{atlanta} & 15 & 22 & 5.28 \\ \hline
     \textbf{dfn-gwin} & 11 & 47 & 13.91 \\ \hline
     \textbf{geant} & 22 & 36 & 23.69 \\ \hline
     \textbf{france} & 25 & 45 & 44.38 \\ \hline
     \textbf{india35} & 35 & 57 & 105.89 \\ \hline
	\end{tabular}
	\caption{Time to solve the edge based LP for various topologies. All values are averaged over 15 runs of CoinLP on a 3.3GHz Intel i5 2500k processor.}
	\label{tab:runtime}
\end{table}

The cost of solving this LP even on small topologies justifies the use of the faster multiplicative weight algorithm instead. The MWU algorithm has a running time of roughly $\tilde{O}(|D| * |E|^2 /\epsilon^2 )$, which on sparse graphs is roughly equal to just the number of variables in the edge-based LP (as opposed to the time needed to actually solve it). While the algorithm is only approximately optimal, choosing an appropriate value of $\epsilon$ (say, $\epsilon = 0.1$) can grant a better running time while still significantly outperforming the naive algorithm.

\section{Middlebox Node Purchase Optimization}
\label{sec:networkdesign}
We now discuss the network design problems mentioned in the introduction. Although such problems can be modeled in multiple ways, we limit our discussion to the case where each vertex $v$ has a potential processing capacity $C$, which can only be utilized if $v$ is ``purchased''. Flow processed elsewhere can be routed through $v$ regardless of whether or not $v$ is purchased. 
\begin{enumerate}
	\item The \textit{minimization} version of the problem (\minmp), where the goal is to pick the smallest set of vertices such that all flow is routable.
	\item The \textit{maximization} version of the problem (\maxmp), where we try to maximize the amount of routable flow while subject to a budget constraint of $k$. 
\end{enumerate} 
Formally, the input to \minmp is a graph $G = (V,E)$, which can be either directed or undirected, with nonnegative costs $q_v$ on its vertices, a potential processing capacity $C : V \rightarrow \nonnegativereal$, and a collection of $(s_i, t_i)$ pairs with demands $R_i$. The goal is to select a set $T \subseteq V$ of vertices such that all demands are satisfied.  \maxmp is given the same collection of inputs along with a budget integer $k$, and the goal is to route as much of the demand as possible.

All four problems (maximization or minimization, directed or undirected), are \nph. We present approximation algorithms and hardness results for each version of the problem, as well as for some restricted variants. Our results are summarized in \Cref{tab:netdesign}.

\begin{figure}
\begin{minipage}{\linewidth}
\captionsetup[table]{aboveskip=2pt}
\captionsetup[table]{belowskip=2pt}
\centering
\captionof{table}{Network Design Results}\label{tab:netdesign} 
\footnotesize
\begin{tabular}{|r|r|c|c|}
\hline
 & & \bf{Directed} & \bf{Undirected} \\ \hline
 \multirow{2}{*}{\bf Budgeted}& \bf{Approximation}& $\Omega(1/\log n)$ & $.078^{(\dagger)}$ \\
 & \bf{Hardness} & $1 - 1/e - \epsilon$ & $.999$ \\ \hline
 \multirow{2}{*}{\bf Minimization}& \bf{Approximation}& $O(\log n)^{(*)}$ & $O(\log n)^{(*)}$ \\
 & \bf{Hardness} & $O(\log n)$ & $2-\epsilon$ \\ \hline
\end{tabular} \\
\normalsize
{\small 
$^{*}$ All demands are satisfied only up to an $(1-\epsilon)$ fraction. \\
$^{\dagger}$ Assuming $1$ source-sink pair. For multiple pairs, we adapt the $\Omega(1/\log n)$-approximation digraph algorithm.
}

\end{minipage}
\end{figure}

\subsection{Bicriterion Approximation Algorithm for Directed and Undirected \minmp}
\label{dirminalg}
We first describe an algorithm for directed \minmp that satisfies all flow requirements up to a factor of  $1-\delta$ fraction with expected cost bounded by $O(\log n / \delta^2)$ times the optimum.

We begin our  approximation algorithm for directed \minmp by modifying the \twalk-based LP formulation with additional 
variables $x_v$ corresponding to whether or not processing capacity at vertex $v$ has been purchased.
We further give a polynomial sized edge-based LP formulation with flow variables $f_i^{1,v}(e)$ and $f_i^{2,v}(e)$
for each commodity $i$, each vertex $v \in V$ and each edge $e \in E$. The variables $f_i^{1,v}(e)$ correspond
to the (processed) commodity $i$ flow that has been processed by vertex $v$: these variables describe a flow 
from $v$ to $t_i$. The variables $f_i^{2,v}(e)$ correspond
to the (unprocessed) commodity $i$ flow that will be processed by vertex $v$: these variables describe a flow 
from $s_i$ to $v$. 

\noindent{}
\begin{minipage}[t]{0.45\textwidth}
\textit{\Twalk-based formulation:}
\footnotesize
 \begin{subequations}
 \footnotesize
\begin{align*}
&\textsc{minimize } \sum_{v \in V} q_v x_v\\
&\textsc{subject to} \\
&x_v \leq 1 &\forall v \in V\\
& p_{i,\pi} = \sum \limits_{v\in \pi} p_{i,\pi}^v& \hspace{-.35in}\forall i \in [|D|], \pi \in P \\
& \sum_{\pi \in P}p_{i,\pi} \geq R_i &\forall i\in[|D|]\\
& \sum\limits_{i=1}^{|D|}\sum \limits_{\stackrel{\pi\in P}{\pi \ni e}} p_{i,\pi} \leq B(e) & \forall e \in E\\
&\sum\limits_{i=1}^{|D|} \sum \limits_{\pi\in P} p_{i, \pi}^v \leq C(v) x_v &\forall v \in V \\
& \sum\limits_{i=1}^{|D|}\sum\limits_{\stackrel{\pi\in P}{\pi \ni e}} p_{i,\pi}^v \leq B(e) x_v &\forall e \in E, v \in V\\
& \sum\limits_{\pi\in P} p_{i,\pi}^v \leq R_i x_v &\hspace{-.2in}\forall i\in[|D|], v \in V,\\
&p_{i,\pi}^v \geq 0 & \hspace{-.55in}\forall i\in[|D|],\pi\in P, v \in \pi\\
&x_v \geq 0 &\forall v \in V
\end{align*}
\end{subequations}

 \end{minipage}
 \hspace{-.22in}
\begin{minipage}[t]{0.55\textwidth}
\textit{\Twalk-based formulation:}
\footnotesize
\begin{subequations}
\footnotesize
\begin{align*}
&\textsc{minimize } \sum_{v \in V} q_v x_v\\
&\textsc{Subject to}\\
&x_v \leq 1 &\forall v \in V\\
&\sum\limits_{e \in \delta^-(u)}  f_i^{j,v}(e)=  \sum\limits_{e \in \delta^+(u)} f_i^{j,v}(e) \hspace{-2in}
& \\
& & \hspace{-1.5in}\forall i \in [|D|], j \in \{1,2\}, v \in V,\forall u \in V \setminus \{s_i,t_i,v\}\\
&\sum\limits_{e \in \delta^-(v)}  f_i^{2,v}(e)=  \sum\limits_{e \in \delta^+(v)} f_i^{1,v}(e) 
&\hspace{-2in}\forall i \in [|D|], v \in V\\
&\sum\limits_{v \in V} \sum\limits_{e \in \delta^+(s_i)} f_i^{2,v}(e) \geq R_i &\forall i \in [|D|]\\
&\sum\limits_{i=1}^{|D|} \sum\limits_{v \in V} (f_i^{1,v}(e) + f_i^{2,v}(e))\leq B(e) &\forall e \in E\\
&\sum\limits_{i=1}^{|D|}  \sum\limits_{e \in \delta^-(v)}  f_i^{2,v}(e) \leq C(v)x_v &\forall v \in V\\
&\sum\limits_{i=1}^{|D|} (f_i^{1,v}(e) + f_i^{2,v}(e)) \leq B(e) x_v \hspace{-2in}&\forall e \in E, v \in V\\
&\sum\limits_{e \in \delta^+(s_i)} f_i^{2,v}(e) \leq R_i x_v &\hspace{-2in}\forall i \in [|D|], v \in V\\
&f_i^{2,v}(e)= 0&\hspace{-2in}\forall i \in [|D|], v \in V, e \in \delta^-(s_i) \\
&f_i^{1,v}(e)= 0&\hspace{-2in}\forall i \in [|D|], v \in V, e \in \delta^+(t_i) \\
&p_i^{1,v}(e), p_i^{2,v}(e), x_v \geq 0&\hspace{-2in}\forall i \in [|D|], v \in V, e \in E 
\end{align*}
\end{subequations}
\normalsize
\end{minipage}
 \ \\

\normalsize
Given an optimal solution to this LP, we pick vertices to install processing capacity on by randomized rounding:
pick vertex $v$ with probability $x_v$. if $x_v$ is picked, then all flows processed by $v$ are rounded up in the
following way: $\hat{F}_i^{j,v}(e) = f_{i}^{j,v}(e)/x_v$ for all $i \in [|D|], j \in \{1,2\}, e \in E$. If $v$ is not picked, then all flows processed by $v$ are set to zero, i.e. $\hat{F}_i^{j,v}(e) = 0$.

By design, $E[\hat{F}_i^{j,v}(e) ] = f_{i}^{j,v}(e)$.
In the solution produced by the rounding algorithm, the total flow through edge $e$ is  
$\displaystyle \sum_{v \in V} \sum_{i=1}^{|D|} ((\hat{F}_i^{1,v}(e) + \hat{F}_i^{2,v}(e))$.
This is a random variable whose expectation is at most $B(e)$, 
and is the sum of independent random variables, one for each vertex $v$.
The constraints of the LP ensure that if $v$ is selected, then the total processing done by vertex $v$ is at most $C(v)$. 
Further, the total contribution of vertex $v$ to the flow on edge $e$ 
does not exceed the capacity $B(e)$, i.e.
$\displaystyle \sum_{i=1}^{|D|} (\hat{F}_i^{1,v}(e) + \hat{F}_i^{2,v}(e)) \leq B(e)$.
Also, the total contribution of vertex $v$ to the commodity $i$ flow is at most $R_i$, i.e.
$\displaystyle \sum_{e \in \delta^+(s_i)} \hat{F}_i^{2,v}(e) \leq R_i$.

We repeat this randomized rounding process $t=O(\log(n)/\epsilon^2)$ times.
Let $g^k(e)$ denote the total flow along edge $e$, and $h^k_i$ denote the total amount of commodity $i$ flow 
in the solution produced by the $k$th round of the randomized rounding process.
The following lemma follows easily by Chernoff-Hoeffding bounds:

\begin{lemma}
\begin{align}
Pr\left[\sum_{k=1}^t g^k(e) \geq (1+\epsilon)t \cdot B(e)\right] &\leq e^{-t \epsilon^2/3} &\forall e \in E\\
Pr\left[\sum_{k=1}^t h^k_i \leq (1-\epsilon) t \cdot R_i \right]  &\leq e^{-t \epsilon^2/2} &\forall i \in [|D|]
\end{align}
\end{lemma}

We set $t = O(\log(n)/\epsilon^2)$ so that the above probabilities are at most $1/n^3$ for each edge $e \in E$ and
each commodity $i$. With high probability, none of the associated events occurs.
The final solution is constructed as follows:
A vertex is purchased if it is selected in any of the $t$ rounds of randomized rounding.
Thus the expected cost of the solution is at most $t = O(\log(n)/\epsilon^2)$ times the LP optimum.
We consider the superposition of all flows produced by the $t$ solutions and scale down the
sum by $t(1+\epsilon)$.
This ensures that the capacity constraints are satisfied.
Note that the vertex processing constraints are also satisfied by the scaled solution.
The total amount of commodity $i$ flow is at least $\frac{1-\epsilon}{1+\epsilon} R_i \geq (1-2\epsilon) R_i$.
Hence we get the following result:
\begin{theorem}
\label{thm:directedminalg}
For directed \minmp, there is a polynomial time randomized algorithm that satisfies all flow requirements up to factor $1-\delta$ and produces 
a solution that respects all capacities, with expected cost bounded by $O(\log(n)/\delta^2)$ times the optimal cost.
\end{theorem}

We can modify the LP to simulate the inclusion of an undirected edge with capacity $B(e)$ by adding the constraints for two arcs between its endpoints with capacity $B(e)$ each, as well as an additional constraint requiring that the sum of flows over these two arcs is bounded by $B(e)$. The analysis done above carries through line-by-line, giving the following result.

\begin{theorem}
\label{thm:undirectedminalg}
For undirected \minmp, there is a polynomial time randomized algorithm that satisfies all flow requirements up to factor $1-\delta$ and produces 
a solution that respects all capacities, with expected cost bounded by $O(\log(n)/\delta^2)$ times the optimal cost.
\end{theorem}

\subsection{Approximation Hardness for Directed \minmp}

\begin{figure}[t]
\centering
 \includegraphics[height=1.0in]{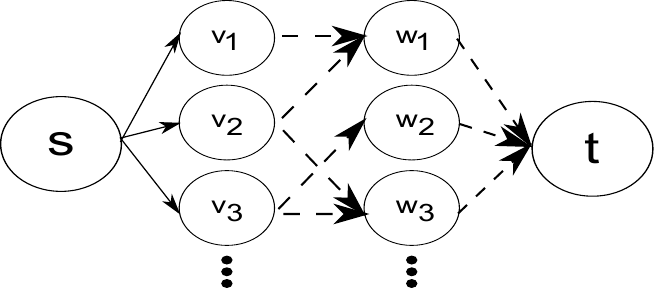}
 \caption{Approximation-preserving reduction from {\sc Set Cover} and {\sc Max $k$-Coverage} to directed \minmp and directed \maxmp. Solid edges have infinite capacity, dashed edges have capacity $1$. $v_i$ vertices have infinite processing potential, at a cost of $1$ each.}
\label{fig:coverhardness}
\end{figure}

We now prove that directed \minmp is NP-hard to approximate to a factor better than $(1-\epsilon)\ln n$ by showing an approximation-preserving reduction from \textsc{Set Cover}, a problem already known to have the aforementioned $(1-\epsilon)\ln n$ hardness \cite{DinurSteurer14}.

Given a \textsc{Set Cover} instance with set system $\mathcal{S} = \{S_1, S_2, \cdots\}$ and universe of elements $\mathcal{U}$, we create one vertex $v_S$ for each $S \in \mathcal{S}$ and one vertex $w_u$ for each $u \in \mathcal{U}$. Further, we create one source vertex $s$ and one sink vertex $t$, where $t$ demands $|\mathcal{U}|$ units of processed flow from $s$. We add one capacity-$n$ arc from $s$ to each $v_S$, and one capacity-$1$ arc from each $w_u$ to $t$. We then add a capacity-$1$ arc from each $v_S$ to $w_u$ whenever $S \ni u$. Finally, we give each $v_S$ vertex $n$ units of processing capacity at a cost of $1$ each.

In order for $t$ to get $|\mathcal{U}|$ units of flow, each $w_u$ must get at least one unit of processed flow itself. Thus, at least one of of its incoming $v_S$ neighbors must be able to process flow. Therefore, this instance of directed \minmp can be seen as the problem of purchasing as few of the $v_S$ vertices so that each $u_W$ vertex has one (or more) incoming $v_S$ vertex. This provides a direct one-to-one mapping between solutions to our constructed instance and the initial \textsc{Set Cover} instance, and the values of the solutions are conserved by the mapping. Therefore, we have an approximation-preserving reduction between the two problems, and directed \minmp acquires the known $(1-\epsilon)\ln n$ inapproximability of \textsc{Set Cover}, summarized in the following result:

\begin{theorem}
For every $\epsilon > 0$, it is \nph to approximate directed \minmp to within a factor of $(1-\epsilon)\ln n$.
\end{theorem}

Note that this construction provides the same hardness even when all demands are only to be satisfied up to a $(1-\delta)$ fraction, showing the asymptotic tightness of the approximation factor in \Cref{thm:directedminalg}.\vspace{-.05in}

\subsection{Approximation Hardness for Undirected \minmp}

We now show an approximation preserving reduction from \textsc{Min Vertex Cover} to undirected \minmp, proving that the latter problem is \ugch to approximate within a factor of $2-\epsilon$ for any $\epsilon>0$ ~\cite{khot2008vertex}, and \nph to approximate within a factor of $1.36$~\cite{dinur2005hardness}. \\

The construction is simple. Given a \textsc{Vertex Cover} instance with graph $G = (V,E)$, we create an identical graph with each vertex $v$ demanding one unit of processed flow from each of its neighbors, and each edge's capacity is $2$. Further, each vertex has $n$ units of processing potential, at a cost of $1$. Because the total demand equals the sum of all edge capacities, each unit of flow sent must use exactly one unit of edge capacity, i.e. all flow paths have length exactly one. Thus, the set of solutions exactly corresponds to vertex covers, with one unit of flow going each way across each edge, from source to sink and either to or from its point of processing. The unit costs ensure that the objective value equals the number of vertices picked, and thus that the optimal solution to this undirected \minmp instance equals that of the original \textsc{Min Vertex Cover}. The conclusion, summarized below, follows.

\begin{theorem}
Approximating undirected \minmp is at least as hard as approximating \textsc{Min Vertex Cover}. In particular, it is \nph to approximate within a factor of $1.36$ and \ugch to approximate within a factor of $2-\epsilon$, for any $\epsilon > 0$.
\end{theorem}

\subsection{Approximation Algorithm for Directed \maxmp}
\label{dirmaxalg}
The algorithm here proceeds similarly to that in \Cref{dirminalg}. The LPs we use are the natural maximization variant of those used for the minimization problem, with the added restriction that we only use a $1/2$ fraction of the budget. It is easy to see that this additional restriction does not reduce the objective value of the optimal LP solution by more than an $1/2$-fraction. We also assume (without loss of generality) that no vertex has cost greater than the budget. The LPs are formulated as follows:\newpage

\noindent{}
\begin{minipage}{0.45\textwidth}
\textit{\Twalk-based formulation:}
\footnotesize
  \begin{subequations}
\begin{align*}
&\textsc{maximize } \sum_{i=1}^{|D|} \sum_{\pi \in P} p_{i, \pi} \\
&\textsc{subject to} \\
&\sum_{v \in V} c_v x_v \leq k/2\\
&x_v \leq 1 &\forall v \in V\\
& p_{i,\pi} = \sum \limits_{v\in \pi} p_{i,\pi}^v& \hspace*{-2in} \forall i \in [|D|], \pi \in P \\
& \sum_{\pi \in P}p_{i,\pi} \geq R_i &\hspace*{-1in}\forall i\in[|D|]\\
& \sum\limits_{i=1}^{|D|}\sum \limits_{\stackrel{\pi\in P}{\pi \ni e}} p_{i,\pi} \leq B(e) \hspace*{-2in}& \forall e \in E\\
&\sum\limits_{i=1}^{|D|} \sum \limits_{\pi\in P} p_{i, \pi}^v \leq C(v) x_v \hspace*{-2in}&\hspace*{-1in}\forall v \in V \\
& \sum\limits_{i=1}^{|D|}\sum\limits_{\stackrel{\pi\in P}{\pi \ni e}} p_{i,\pi}^v \leq B(e) x_v \hspace*{-2in}&\forall e \in E, v \in V\\
& \sum\limits_{\pi\in P} p_{i,\pi}^v \leq R_i x_v &\hspace*{-2in}\forall i\in[|D|], v \in V,\\
&p_{i,\pi}^v \geq 0 & \hspace*{-1in} \forall i\in[|D|],\pi\in P, v \in \pi\\
&0 \leq x_v \leq 1 &\forall v \in V
\end{align*}
\end{subequations}
\normalsize
\end{minipage}

\textit{Edge-based formulation:}
\footnotesize
\begin{subequations}
\begin{align*}
&\textsc{maximize }\sum_{v \in V}\sum\limits_{i=1}^{|D|}\sum \limits_{e\in \delta^-(v)} f_i^{2,v}(e)  \\
&\textsc{Subject to}\\
&\sum_{v \in V} c_v x_v \leq k/2\\
&\sum\limits_{e \in \delta^-(u)}  f_i^{j,v}(e)=  \sum\limits_{e \in \delta^+(u)} f_i^{j,v}(e) \hspace*{-2.5in}
&\\
& & \hspace*{-2in}\forall i \in [|D|], j \in \{1,2\}, v \in V, \forall u \in V \setminus \{s_i,t_i,v\} \\
&\sum\limits_{e \in \delta^-(v)}  f_i^{2,v}(e)=  \sum\limits_{e \in \delta^+(v)} f_i^{1,v}(e) \hspace*{-2in}
&\forall i \in [|D|], v \in V,\\
&\sum\limits_{v \in V} \sum\limits_{e \in \delta^+(s_i)} f_i^{2,v}(e) \geq R_i &\forall i \in [|D|]\\
&\sum\limits_{i=1}^{|D|} \sum\limits_{v \in V} (f_i^{1,v}(e) + f_i^{2,v}(e))\leq B(e) &\forall e \in E\\
&\sum\limits_{i=1}^{|D|}  \sum\limits_{e \in \delta^-(v)}  f_i^{2,v}(e) \leq C(v)x_v &\hspace*{-2in}\forall v \in V\\
&\sum\limits_{i=1}^{|D|} (f_i^{1,v}(e) + f_i^{2,v}(e)) \leq B(e) x_v &\hspace*{-2in}\forall e \in E, v \in V\\
&\sum\limits_{e \in \delta^+(s_i)} f_i^{2,v}(e) \leq R_i x_v \hspace*{-2in}&\hspace*{-2in}\forall i \in [|D|], v \in V\\
&f_i^{2,v}(e)= 0&\hspace*{-2in}\forall i \in [|D|], v \in V, e \in \delta^-(s_i) \\
&f_i^{1,v}(e)= 0&\hspace*{-2in}\forall i \in [|D|], v \in V, e \in \delta^+(t_i) \\
&p_i^{1,v}(e), p_i^{2,v}(e), x_v \geq 0&\hspace*{-2in}\forall i \in [|D|], v \in V, e \in E \\
&0 \leq x_v \leq 1 &\hspace*{-2in}\forall v \in V \\
\end{align*}
\end{subequations}

\normalsize
If purchasing a single vertex allows us to route a $1 / (2 \ln n)$ fraction of the objective value of the above LP, we purchase only this vertex. Otherwise, we can remove the potential for processing at each vertex $v$ with $c_v \geq k / \ln n$ and re-solve the LP to get a solution with objective value at least half as large as before. Thus, from now on we can assume that no $c_v$ exceeds $k / \ln n$ and therefore that the optimal LP solution puts support on at least a $1 / \ln n$ fraction of the $x_v$s (at a cost of $2$ in our approximation factor). We will call the objective value of this modified linear program $\mathrm{OPT}_\mathrm{LP'}$.

Again, we pick the vertices on which to install processing capacity on by randomized rounding: each vertex $v$ is picked with probability $x_v$. If $x_v$ is picked, then all flows processed by $v$ are rounded so that $\hat{F}_i^{j,v}(e) = f_{i}^{j,v}(e)/(4 x_v \ln n)$ for all $i \in [|D|], j \in \{1,2\}, e \in E$. If $v$ is not picked, then all flows processed by $v$ are set to zero, i.e. $\hat{F}_i^{j,v}(e) = 0$.

By design, $E[\hat{F}_i^{j,v}(e) ] = f_{i}^{j,v}(e) /(4 \ln n)$ and thus the total amount of flow processed, $P$, satisfies $E[P] = E\left[\sum\limits_{v \in V}\sum\limits_{i=1}^{|D|}\sum \limits_{e\in \delta^-(v)} \hat{F}_i^{2,v}(e)\right] = \mathrm{OPT}_\mathrm{LP'} / (4 \ln n)$.
In the solution produced by the rounding algorithm, the total flow through edge $e$ is  
$\displaystyle\sum_{v \in V} \sum_{i=1}^{|D|} ((\hat{F}_i^{1,v}(e) + \hat{F}_i^{2,v}(e))$.
This sum of random variables is $\hat{B}(e) = B(e)/(4 \ln n)$ in expectation. Letting $g(e)$ denote the flow along edge $e$, standard bounds give
\footnotesize
\begin{lemma}
\label{lem:maxround}
\begin{align}
Pr\left[g(e) \geq (4 \lg n) \cdot \hat{B}(e)\right] &\leq e^{-4 \ln n} = n^{-4} &\forall e \in E \\
Pr\left[P \leq (1/4) \cdot (1/(4\lg n) \cdot \mathrm{OPT}_\mathrm{LP'}) \right] &\leq e^{-4 \ln n} = n^{-4} &\forall e \in E
\end{align}
\end{lemma}
\normalsize
so by the union bound, with probability higher than $1 - 1/n$ every edge is assigned $\leq B(e)$ total flow and the amount of flow processed and routed is within a $1/16 \ln n$ factor of $\mathrm{OPT}_\mathrm{LP'}$. \\

Finally, by Markov's inequality, the original budget constraint is satisfied with probability at least $1/2$. Combining this with lemma \Cref{lem:maxround}, the algorithm fails with probability at most $1/2 + 1/n$. Repeating the algorithm $O(\log n)$ times and taking the best feasible solution therefore provides an $\Omega(1/\log n)$ approximation with probability at least $1 - 1/\poly(n)$. This can be summarized in the following result:

\begin{theorem}
For directed \maxmp, there is a polynomial-time randomized algorithm producing an $\Omega(1/\log(n))$ approximation.
\end{theorem}

We can also apply this algorithm to undirected instances by adding additional constraints the as we did in \Cref{dirminalg}, with the analysis carrying through as before. Thus, we attain the following:

\begin{theorem}
For undirected \maxmp, there is a polynomial-time randomized algorithm producing an $\Omega(1/\log(n))$ approximation.
\end{theorem}
\subsection{Approximation Algorithm for Undirected \maxmp}

We now show that the undirected \maxmp admits a constant-factor approximation algorithm when restricted to a single source $s$.  Let $\mathrm{OPT}(G,k)$ denote the value of the optimal solution to an instance with graph $G$ and budget $k$. Our algorithm works by splitting the problem into both a \textit{processing step} and a \textit{routing step}. The algorithm begins by reserving a $1/2$ fraction of each edge for use in the processing step and the remaining $1/2$ fraction for use in the routing step. Calling the reserved-capacity graphs $G_{\textrm{proc}}$ and $G_{\textrm{route}}$, respectively, the algorithm proceeds as follows:

\paragraph{Processing step}

A well known fact in capacitated network design is that the maximum amount of flow routable (sans processing) from a set $S \subseteq V$ of source vertices to a single sink forms a monotone, submodular function in $S$~\cite{chakrabarty2013capacitated}. Although this problem is usually defined in the context of sources that can produce an arbitrary amount of flow (should the network support it), we can bottleneck each source $s_i$ into producing at most some $c_i$ units of flow by replacing it with a pair of vertices connected by a capacity $c_i$ edge, without changing the submodularity of the routable flow function, $f_G(S)$. For the purpose of this lemma, redefining $s$ as our ``sink'' and the set $P$ of processing nodes as our source set $S$, we immediately attain that the function $f_G(P)$ is submodular, where $P \subset V$ is the set of nodes purchased for processing. \\

Let $H$ be a copy of $G_{\mathrm{proc}}$ with all edge capacities halved. Because $f_H$ is a submodular function, the problem of using our budget to purchase a set $P \subseteq V$ of processing nodes so to maximize $f_H(P)$ is simply an instance of a monotone, submodular maximization subject to knapsack constraints. Such problems are known to admit simple $(1 - 1/e)$-approximation algorithms~\cite{sviridenko2004note}. Let $P(H,k)$ be the optimal solution to this \textit{processable flow problem} on $H$ with budget $k$ and $\mathrm{ALG}_1(H,K)$ denote the value of the solution found by our algorithm. Because $P(H,k)$ is an upper bound on $\mathrm{OPT}({H,k})$ (indeed, the former is simply an instance of the former without the need to account for post-processing routing), the $(1-1/e)$ approximation we get has value at least equal to $(1-1/e)$ times the value of $\mathrm{OPT}({H,k})$. In particular

\begin{align*}
\mathrm{ALG}_1(H,k) & \geq (1-1/e) P(H,k) \\
                     & \geq (1-1/e)\mathrm{OPT}(H,k) \\
                     & \geq (1-1/e)(1/2)\mathrm{OPT}(G_{\mathrm{proc}},k) \\
                     & \geq (1-1/e)(1/2)(1/2)\mathrm{OPT}(G,k) \\
                     & = (1 - 1/e)/4 \cdot \mathrm{OPT}(G,k)
\end{align*}

Further, because our solution only uses at most half of the capacity of any edge in $G_{\mathrm{proc}}$, we can use the remaining, unused half of the capacities to route all flow we managed to process back to $s$.

\paragraph{Routing Step}

All flow residing in $s$ after the end of the processing step is already processed, all of it can be routed directly to the sinks using the $1/2$ fraction of edge capacities we reserved for $G_{\mathrm{route}}$. Because multiplying all edge capacities by $1/2$ reduces the amount of routable flow by the same (multiplicative) amount, we can route at least $(1/2)\min(\mathrm{ALG}_1(H,k), \textsc{MaxFlow}_G(s,t))$ units of the processed flow from $s$ to $t$. As $\textsc{MaxFlow}_G(s,t)$ is a (trivial) upper bound on $\mathrm{OPT}(G,k)$, this means we can route at least $(1/2)(1-1/e)/4\mathrm{OPT}(G,k)$ units of the processed flow from $s$ to the sinks, giving a $(1-1/e)/8 > .078$ approximation algorithm.

Thus, we get the following theorem:

\begin{theorem}
For undirected \maxmp with a single source, there is a deterministic polynomial time algorithm that produces 
a solution that can route at least $(1 - 1/e)/8 \approx .078$ times the optimal objective solution.
\end{theorem}
\subsection{Approximation Hardness for Directed \maxmp}

We now prove that directed \maxmp is NP-hard to approximate to a factor better than $(1 - 1/e + \epsilon)$. To show this, we reduce from \textsc{Max k-Cover}, which is known to have the same hardness result \cite{Feige98}. \\

Given a \textsc{Max k-Cover} instance with set system $\mathcal{S}$ and universe of elements $\mathcal{U}$, we create one vertex $v_S$ for each $S \in \mathcal{S}$ and one vertex $w_u$ for each $u \in \mathcal{U}$. Further, we create one source vertex $s$ and one sink vertex $t$, where $t$ demands $|\mathcal{U}|$ units of processed flow from $s$. We add one capacity-$n$ arc from $s$ to each $v_S$, and one capacity-$1$ arc from each $w_u$ to $t$. We then add a capacity-$1$ arc from each $v_S$ to $w_u$ whenever $S \ni u$. Finally, we give each $v_S$ vertex $n$ units of processing capacity at a cost of $1$ each. The budget for the instance is $k$ -- the same as the budget for the \textsc{Max-k-Cover} instance. A diagram of the reduction is given in \Cref{fig:coverhardness}.\\

When flow is routed maximally, each $w_u$ contributes $1$ unit of flow to the total $s-t$ flow if and only if it has a neighbor $v_S$ that was chosen to be active. Otherwise, this vertex does not help contribute towards the $s-t$ flow. Thus, this instance of directed \maxmp can be seen as the problem of buying $k$ different $v_S$ vertices so to maximize the number of distinct $w_u$ vertices to which they are adjacent. Thus, there is a direct one-to-one mapping between solutions to our constructed instance and the initial \textsc{Max k-Cover} instance, and the values of the solutions are conserved by the mapping. Therefore, we have an approximation-preserving reduction between the two problems, and directed \maxmp acquires the known $(1 - 1/e + \epsilon)$ inapproximability of \textsc{Max k-Cover}.

\subsection{Approximation Hardness for Undirected \maxmp}
\newcommand{\bopt}{b_{\rm OPT}\xspace}
\newcommand{\mb}{{\sc Max Bisection}\xspace}

We show that for some fixed $\epsilon_0 > 0$, the undirected version of \maxmp is \nph to approximate within a factor of $1-\epsilon$, implying that the the problem does not admit a PTAS unless $\p = \np$. We make no attempt to maximize the value $\epsilon_0$. \\

We show this hardness by reducing from \mb on degree-3 graphs, shown to be hard to approximate within a factor of $.997$ in \cite{berman1999some}\footnote{To be precise, this paper shows the aforementioned hardness for \textsc{Max Cut}. A simple approximation preserving reduction from \textsc{Max Cut} to \mb can be derived by looking at maximum cuts of the graph formed by 2 disjoint copies of the \textsc{Max Cut} instance graph.}. Let $G = (V, E)$ be the input to the degree-3 \mb instance. For each $v_i \in V$, create two vertices, $u_i$ and $w_i$, joined by an edge with capacity $3$. We also add a capacity-$1$ edge between $u_i$ and $u_j$ whenever $v_i$ and $v_j$ are adjacent in $G$. Each $w_i$ vertex demands $3$ units of flow from every $u_j$ (including when $i = j$). Further, every $u_i$ vertex can be given $3|V|$ units of processing capacity (or, equivalently, $\infty$ units) at a cost of $1$, and the instance's budget is set to $|V|/2$. \\

The intuition behind the construction is as follows. With a budget of $|V|/2$, we can purchase exactly half of the $u_i$ vertices (and all budget is used up without loss of generality); our bisection will be between the purchased $u_i$s and the unpurchased ones. Let $b$ be the number of edges in any such bisection. Each $w_i$ adjacent to a purchased $u_i$ can have $3$ units of its demand satisfied by flow originating from and processed by $u_i$, and the only edge connecting $w_i$ to the rest of the graph ensures $w_i$ can never receive more than $3$ units of flow regardless. Thus, such $w_i$s are maximally satisfied, and contribute $3|V|/2$ units to our objective value. The remaining $w_i$s must have their processed flow routed to them via edge via the $b$ capacity-$1$ edges in the bisection (and, indeed, every edge in the bisection will carry $1$ unit of flow when routed optimally, as witnessed by the solution where each unprocessed $u_i$ receives flow on each cut-edge and routes it directly to $w_i$), so the total amount of demand satisfied by the $w_i$ adjacent to unpurchased vertices is exactly $b$, so the objective value of a solution with $b$ edges in the bisection is exactly $3|V|/2 + b$. \\

Let $\bopt$ denote the number of edges cut by the optimal bisection. It is a well-known fact that $\bopt \geq |E|/2 = 3|V|/4$. By the theorem of \cite{berman1999some} it is \nph to distinguish instances with $3|V|/2 + \bopt$ units of satisfiable demand from those with only $3|V|/2 + (1-.003)\bopt$, giving an inapproximability ratio of

\begin{align*}
\frac{3|V|/2 + (1-.003)\bopt}{3|V|/2 + \bopt} & = 1 - \frac{.003\bopt}{3|V|/2 + \bopt} \\
& = 1 - \frac{.003}{3|V|/(2\bopt) + 1} \\
& \leq 1 - \frac{.003}{3|V|/(2\cdot3|V|/4) + 1} \\
& = 1 - \frac{.003}{2+1} \\
& = .999
\end{align*}

This calculation is summarized in the following result:

\begin{theorem}
It is \nph to approximate undirected \maxmp to within a factor better than $.999$.
\end{theorem}

\section{Related Work}

\subparagraph{Network Function Optimization}
In software-defined networking, SIMPLE~\cite{SIMPLE2013} and FlowTags~\cite{FlowTags} take advantage of switches with fine-grained rule support. Both approaches focus on how to utilize the constrained TCAM size, a hardware limitation to support fine-grained policy. Neither approach attempts to solve the joint optimization of the capacity constraints for both servers and switches. Slick~\cite{slick} offers a high-level control program that specifies custom processing on precise subset of flows.  It also assumes the server processing power is heterogeneous, and uses heuristic approaches for the underlying placement, routing and steering. 

\subparagraph{Network Function Consolidation}
CoMB~\cite{COMB2012} and Click~\cite{clickos} both consolidate network functions into applications or a VM images, and consider server machines that can each run multiple instances of different network functions. Both focus on improving the performance on single nodes, and treat network functions homogeneously. Neither covers a network-wide optimization.

\subparagraph{Network Function Migration and Reroute}
OpenNF~\cite{OpenNF} and Split-Merge~\cite{SplitMerge} leverage the SDN controller to manage the network function's state migration and the network function's flow migration. Both focus on reallocating resources and rerouting flows when either a node or a link is over-utilized. While their solution focuses on fixing congestion when it occurs, ours focuses on figuring out how to avoid congestion in the first place.

\subparagraph{Network Function Online Request Model}
Recently, Even, Medina, and Patt-Shamir \cite{even2016competitive} studied an online request admission problem in the same multi-commodity flow with processing setting that we study. In their work, requests arrive online and specify a processing pipeline for flow between a source and sink; intermediate nodes in the pipeline may be any subset of nodes in the underlying graph. The goal is to accept as many such flow requests as possible while ensuring that accepted requests are assigned flow paths that satisfy capacity constraints. In this setting, the authors show an $O(k \log (kn))$-competitive online algorithm for instances with length-$k$ pipelines.

\subparagraph{Routing and Middlebox optimization}
A couple of recent papers consider approximation algorithms for path computation and service placement \cite{even2016approximation} and Service Chain and Virtual Network Embeddings \cite{rost2016service}. Both papers use randomized rounding of a linear programming relaxation of the problem. Both of these works differ from our paper in that packets between demand pairs are not splittable, and thus must be sent along \emph{paths} rather than \emph{flows}. Other recent papers provide approximation algorithms for variants of \minmp with no hard edge constraints~\cite{cohen2015near,lukovszki2017approximate}. In~\cite{lukovszki2017approximate}, the authors independently derive the same \textsc{Set Cover}-based hardness construction for their problem variant.

\bibliography{references}

\begin{thebibliography}{10}

\bibitem{slick}
Bilal Anwer, Theophilus Benson, Nick Feamster, and Dave Levin.
\newblock Programming {Slick} network functions.
\newblock In {\em Proceedings of Symposium on SDN Research}, June 2015.

\bibitem{mwu}
Sanjeev Arora, Elad Hazan, and Satyen Kale.
\newblock The multiplicative weights update method: A meta-algorithm and
  applications.
\newblock {\em Theory of Computing}, 8(1):121--164, 2012.

\bibitem{berman1999some}
Piotr Berman and Marek Karpinski.
\newblock {\em On Some Tighter Inapproximability Results}.
\newblock Springer, 1999.

\bibitem{chakrabarty2013capacitated}
Deeparnab Chakrabarty, Ravishankar Krishnaswamy, Shi Li, and Srivatsan
  Narayanan.
\newblock Capacitated network design on undirected graphs.
\newblock In {\em Approximation, Randomization, and Combinatorial Optimization.
  Algorithms and Techniques}, pages 71--80. Springer, 2013.

\bibitem{nfv}
M~Chiosi et~al.
\newblock {Network Functions Virtualisation}: Introductory white paper.
\newblock In {\em {SDN} and {OpenFlow} World Congress}, Oct 2012.

\bibitem{cohen2015near}
Rami Cohen, Liane Lewin-Eytan, Joseph~Seffi Naor, and Danny Raz.
\newblock Near optimal placement of virtual network functions.
\newblock In {\em Computer Communications (INFOCOM), 2015 IEEE Conference on},
  pages 1346--1354. IEEE, 2015.

\bibitem{dinur2005hardness}
Irit Dinur and Samuel Safra.
\newblock On the hardness of approximating minimum vertex cover.
\newblock {\em Annals of Mathematics}, pages 439--485, 2005.

\bibitem{DinurSteurer14}
Irit Dinur and David Steurer.
\newblock Analytical approach to parallel repetition.
\newblock In {\em Proceedings of the Annual ACM Symposium on Theory of
  Computing}, pages 624--633, New York, NY, USA, 2014. ACM.
\newblock URL: \url{http://doi.acm.org/10.1145/2591796.2591884}, \href
  {http://dx.doi.org/10.1145/2591796.2591884}
  {\path{doi:10.1145/2591796.2591884}}.

\bibitem{even2016competitive}
Guy Even, Moti Medina, and Boaz Patt-Shamir.
\newblock Competitive path computation and function placement in sdns.
\newblock {\em arXiv preprint arXiv:1602.06169}, 2016.

\bibitem{even2016approximation}
Guy Even, Matthias Rost, and Stefan Schmid.
\newblock An approximation algorithm for path computation and function
  placement in sdns.
\newblock In {\em Proceedings of SIROCCO}, 2016.

\bibitem{FlowTags}
Seyed~Kaveh Fayazbakhsh, Luis Chiang, Vyas Sekar, Minlan Yu, and Jeffrey~C.
  Mogul.
\newblock Enforcing network-wide policies in the presence of dynamic middlebox
  actions using flowtags.
\newblock In {\em 11th USENIX Symposium on Networked Systems Design and
  Implementation (NSDI 14)}, pages 543--546, Seattle, WA, April 2014. USENIX
  Association.
\newblock URL:
  \url{https://www.usenix.org/conference/nsdi14/technical-sessions/presentation/fayazbakhsh}.

\bibitem{Feige98}
Uriel Feige.
\newblock A threshold of ln n for approximating set cover.
\newblock {\em Journal of the ACM (JACM)}, 45(4):634--652, 1998.

\bibitem{OpenNF}
Aaron Gember-Jacobson, Raajay Viswanathan, Chaithan Prakash, Robert Grandl,
  Junaid Khalid, Sourav Das, and Aditya Akella.
\newblock {OpenNF}: Enabling innovation in network function control.
\newblock In {\em Proceedings of the ACM Conference on SIGCOMM}, pages
  163--174. ACM, 2014.
\newblock URL: \url{http://doi.acm.org/10.1145/2619239.2626313}, \href
  {http://dx.doi.org/10.1145/2619239.2626313}
  {\path{doi:10.1145/2619239.2626313}}.

\bibitem{heorhiadi2015accelerating}
Victor Heorhiadi, Michael~K Reiter, and Vyas Sekar.
\newblock Accelerating the development of software-defined network optimization
  applications using sol.
\newblock {\em arXiv preprint arXiv:1504.07704}, 2015.

\bibitem{NSDI2016}
Victor Heorhiadi, Michael~K. Reiter, and Vyas Sekar.
\newblock Simplifying software-defined network optimization using sol.
\newblock In {\em 13th USENIX Symposium on Networked Systems Design and
  Implementation (NSDI 16)}, pages 223--237, Santa Clara, CA, March 2016.
  USENIX Association.
\newblock URL:
  \url{https://www.usenix.org/conference/nsdi16/technical-sessions/presentation/heorhiadi}.

\bibitem{jointoptvideo2015}
Y.~Jin, Y.~Wen, and C.~Westphal.
\newblock Towards joint resource allocation and routing to optimize video
  distribution over future internet.
\newblock In {\em IFIP Networking Conference (IFIP Networking), 2015}, pages
  1--9, May 2015.
\newblock \href {http://dx.doi.org/10.1109/IFIPNetworking.2015.7145311}
  {\path{doi:10.1109/IFIPNetworking.2015.7145311}}.

\bibitem{khot2008vertex}
Subhash Khot and Oded Regev.
\newblock Vertex cover might be hard to approximate to within 2- $\varepsilon$.
\newblock {\em Journal of Computer and System Sciences}, 74(3):335--349, 2008.

\bibitem{LiQian2016}
Xin Li and Chen Qian.
\newblock A survey of network function placement.
\newblock In {\em 2016 13th IEEE Annual Consumer Communications Networking
  Conference (CCNC)}, pages 948--953, Jan 2016.
\newblock \href {http://dx.doi.org/10.1109/CCNC.2016.7444915}
  {\path{doi:10.1109/CCNC.2016.7444915}}.

\bibitem{lukovszki2017approximate}
Tamas Lukovszki, Matthias Rost, and Stefan Schmid.
\newblock Approximate and incremental network function placement.
\newblock {\em arXiv preprint arXiv:1706.06496}, 2017.

\bibitem{clickos}
Joao Martins, Mohamed Ahmed, Costin Raiciu, Vladimir Olteanu, Michio Honda,
  Roberto Bifulco, and Felipe Huici.
\newblock Clickos and the art of network function virtualization.
\newblock In {\em 11th USENIX Symposium on Networked Systems Design and
  Implementation (NSDI 14)}, pages 459--473. USENIX Association, April 2014.
\newblock URL:
  \url{https://www.usenix.org/conference/nsdi14/technical-sessions/presentation/martins}.

\bibitem{opnfv}
{OPNFV}.
\newblock Opnfv: An open platform to accelerate nfv.
\newblock {Linux Foundation}, \url{https://www.opnfv.org/}.

\bibitem{orlowski2010sndlib}
Sebastian Orlowski, Roland Wess{\"a}ly, Michal Pi{\'o}ro, and Artur
  Tomaszewski.
\newblock Sndlib 1.0—survivable network design library.
\newblock {\em Networks}, 55(3):276--286, 2010.

\bibitem{pst}
Serge~A. Plotkin, David~B. Shmoys, and E.~Tardos.
\newblock Fast approximation algorithms for fractional packing and covering
  problems.
\newblock In {\em Proceedings of the Symposium on the Foundations of Computer
  Science}, pages 495--504, Oct 1991.
\newblock \href {http://dx.doi.org/10.1109/SFCS.1991.185411}
  {\path{doi:10.1109/SFCS.1991.185411}}.

\bibitem{SIMPLE2013}
Zafar~Ayyub Qazi, Cheng-Chun Tu, Luis Chiang, Rui Miao, Vyas Sekar, and Minlan
  Yu.
\newblock {SIMPLE-fying Middlebox Policy Enforcement Using SDN}.
\newblock In {\em Proceedings of ACM SIGCOMM}, pages 27--38. ACM, 2013.
\newblock URL: \url{http://doi.acm.org/10.1145/2486001.2486022}, \href
  {http://dx.doi.org/10.1145/2486001.2486022}
  {\path{doi:10.1145/2486001.2486022}}.

\bibitem{SplitMerge}
Shriram Rajagopalan, Dan Williams, Hani Jamjoom, and Andrew Warfield.
\newblock Split/merge: System support for elastic execution in virtual
  middleboxes.
\newblock In {\em Presented as part of the 10th USENIX Symposium on Networked
  Systems Design and Implementation (NSDI 13)}, pages 227--240, Lombard, IL,
  2013. USENIX.
\newblock URL:
  \url{https://www.usenix.org/conference/nsdi13/technical-sessions/presentation/rajagopalan}.

\bibitem{rost2016service}
Matthias Rost and Stefan Schmid.
\newblock Service chain and virtual network embeddings: Approximations using
  randomized rounding.
\newblock {\em arXiv preprint arXiv:1604.02180}, 2016.

\bibitem{COMB2012}
Vyas Sekar, Norbert Egi, Sylvia Ratnasamy, Michael~K. Reiter, and Guangyu Shi.
\newblock Design and implementation of a consolidated middlebox architecture.
\newblock In {\em Proceedings of the 9th USENIX Conference on Networked Systems
  Design and Implementation}, NSDI'12, pages 24--24. USENIX Association, 2012.
\newblock URL: \url{http://dl.acm.org/citation.cfm?id=2228298.2228331}.

\bibitem{APLOMB2012}
Justine Sherry, Shaddi Hasan, Colin Scott, Arvind Krishnamurthy, Sylvia
  Ratnasamy, and Vyas Sekar.
\newblock Making middleboxes someone else's problem: Network processing as a
  cloud service.
\newblock In {\em Proceedings of the ACM SIGCOMM 2012 Conference on
  Applications, Technologies, Architectures, and Protocols for Computer
  Communication}, SIGCOMM '12, pages 13--24. ACM, 2012.
\newblock URL: \url{http://doi.acm.org/10.1145/2342356.2342359}, \href
  {http://dx.doi.org/10.1145/2342356.2342359}
  {\path{doi:10.1145/2342356.2342359}}.

\bibitem{sviridenko2004note}
Maxim Sviridenko.
\newblock A note on maximizing a submodular set function subject to a knapsack
  constraint.
\newblock {\em Operations Research Letters}, 32(1):41--43, 2004.

\bibitem{uhlig2006providing}
Steve Uhlig, Bruno Quoitin, Jean Lepropre, and Simon Balon.
\newblock Providing public intradomain traffic matrices to the research
  community.
\newblock {\em ACM SIGCOMM Computer Communication Review}, 36(1):83--86, 2006.

\end{thebibliography}

\end{document}